\documentclass[sigconf]{acmart}
\usepackage{algorithm}
\usepackage{algorithmic}
\usepackage{cryptocode}
\usepackage{multirow}
\usepackage{threeparttable}

\usepackage{subfig}

\usepackage{colortbl}

\usepackage{amsmath,amsfonts}
\usepackage{amsthm}

\usepackage{graphicx}
\usepackage{tikz}
\usepackage{xspace}

\usepackage{thmtools}
\usepackage{thm-restate}

\definecolor{ao}{rgb}{0.0,0.65,0.0}
\definecolor{lightred}{rgb}{0.95,0.45,0.5}

\newcommand{\tikzcircle}[2][red,fill=red]{\tikz[baseline=-0.5ex]\draw[#1,radius=#2] (0,0) circle ;}

\def\ie{\textit{i.e.}\xspace} 

\def\eg{\textit{e.g.}\xspace}
\def\etal{\textit{et~al.}\xspace}
\def\cf{\textit{cf.}\xspace}

\definecolor{grayL}{RGB}{255,242,248}
\definecolor{colora}{RGB}{147,112,219}
\definecolor{colorb}{RGB}{219,112,147}

\def\textcgnn{VIRGOS}
\def\cgnn{\textnormal{\textsc{Virgos}}\xspace}
\def\osmm{\textnormal{(SM)$^2$}\xspace}

\def\oom{\mathbin{\text{\textcircled{M}}}}

\def\adjmat{\mathsf{A}}

\def\weimat{\mathsf{W}}
\def\feamat{\mathsf{X}}

\def\adjout{\adjmat_{\mathsf{out}}}
\def\adjin{\adjmat_{\mathsf{in}}}

\def\graph{\mathcal{G}}
\def\graphout{\mathcal{G}_{\mathsf{out}}}
\def\graphin{\mathcal{G}_{\mathsf{in}}}
\def\vertex{\mathcal{V}}
\def\edge{\mathcal{E}}
\def\numnode{|\vertex|}
\def\numedge{|\edge|}
\def\numnonzero{t}

\def\ssp{\mathsf{OP}}
\def\SM{\mathsf{OSM}}

\def\promult{\Pi_{\mathsf{Mult}}}
\def\smm{\mathsf{(SM)^2}}
\def\prosmm{\Pi_{\smm}}

\def\l{\langle}
\def\r{\rangle}

\def\type{\mathsf{type}}
\def\raw{\mathsf{plain}}
\def\shared{\mathsf{shared}}

\def\bvec{\mathsf{B}}
\def\uvec{\mathsf{U}}
\def\xvec{\mathsf{X}}

\def\bitlen{L} 

\def\view{\mathsf{view}}

\def\ideal{\mathsf{Ideal}}
\def\real{\mathsf{Real}}
\def\Scal{\mathcal{S}}
\def\view{\mathsf{view}}
\def\output{\mathsf{output}}
\def\hyb{\mathsf{Hyb}}

\def\A{\mathcal{A}}
\def\F{\mathcal{F}}

\def\pp{\mathcal{P}}
\def\tt{\mathcal{T}}

\def\Rcal{\mathcal{R}}

\def\Abb{\mathbb{A}}
\def\Bbb{\mathbb{B}}
\def\Mbb{\mathbb{M}}
\def\Sbb{\mathbb{S}}
\def\Zbb{\mathbb{Z}}

\def\Asf{\mathsf{A}}

\def\Msf{\mathsf{M}}
\def\Psf{\mathsf{P}}
\def\Qsf{\mathsf{Q}}
\def\Xsf{\mathsf{X}}
\def\Ysf{\mathsf{Y}}
\def\Zsf{\mathsf{Z}}

\def\prf{\mathsf{PRF}}
\def\key{\mathsf{key}}
\def\ctr{\mathsf{ctr}}

\def\nrow{n_{\mathsf{row}}}
\def\ncol{n_{\mathsf{col}}} 
\def\gammaout{\Gamma_{\mathsf{out}}}
\def\gammain{\Gamma_{\mathsf{in}}} 

\newtheorem{theorem}{Theorem}
\newtheorem{definition}{Definition}

\newenvironment{protocol}[1][htb]{%
	\floatname{algorithm}{Protocol}
	\begin{algorithm}[#1]%
	}{\end{algorithm}}

\newenvironment{functionality}[1][htb]{%
	\floatname{algorithm}{Functionality}
	\begin{algorithm}[#1]%
	}{\end{algorithm}}

    \setcopyright{acmlicensed}
\copyrightyear{2025}
\acmYear{2025}
\acmDOI{XXXXXXX.XXXXXXX}
\acmConference[Conference acronym 'XX]{Make sure to enter the correct
  conference title from your rights confirmation email}{June 03--05,
  2025}{Woodstock, NY}
\acmISBN{978-1-4503-XXXX-X/18/06}

\AtBeginDocument{%
  }

\acmSubmissionID{123-A56-BU3}

\begin{document}
\title{$\cgnn$: Secure Graph Convolutional Network on Vertically Split Data from Sparse Matrix Decomposition}
 
\author{Yu Zheng}
\authornote{Yu and Qizhi share the co-first authorship.}
\email{yu.zheng@uci.edu}
 \affiliation{%
 University of California, Irvine
  \country{ }
}

\author{Qizhi Zhang}
\authornotemark[1]
\email{zqz.math@gmail.com}
\affiliation{%
Morse Team, Ant Group
  \country{ }
}

\author{Lichun Li}
\email{lichun.llc@antgroup.com}
\affiliation{%
Morse Team, Ant Group
  \country{ }
}

\author{Kai Zhou}
\email{kaizhou@polyu.edu.hk}
\affiliation{%
Hong Kong Polytechnic University
  \country{ }
}

\author{Shan Yin}
\email{yinshan.ys@antgroup.com}
\affiliation{%
Morse Team, Ant Group
  \country{ }
}
 
\begin{abstract}
Securely computing graph convolutional networks (GCNs) is critical for applying their analytical capabilities to privacy-sensitive data like social/credit networks. 
Multiplying a sparse yet large adjacency matrix of a graph in GCN---a core operation in training/inference---poses a performance bottleneck in secure GCNs. 
Consider a GCN with $\numnode$ nodes and $\numedge$ edges; it incurs a large $O(\numnode^2)$ communication overhead.

Modeling bipartite graphs and leveraging the monotonicity of non-zero entry locations, we propose a co-design harmonizing secure multi-party computation (MPC) with matrix sparsity.
Our sparse matrix decomposition transforms an arbitrary sparse matrix into a product of structured matrices.
Specialized MPC protocols for oblivious permutation and selection multiplication are then tailored, enabling our secure sparse matrix multiplication (\osmm) protocol, optimized for secure multiplication of these structured matrices.
Together, these techniques take $O(\numedge)$ communication in constant rounds.
Supported by \osmm, we present \cgnn\footnote{
Vertically-split Inference \& Reasoning on GCNs Optimized by Sparsity. 
}, a secure $2$-party framework that is communication-efficient and memory-friendly on standard vertically-partitioned graph datasets. 
Performance of \cgnn has been empirically validated across diverse network conditions. 
\end{abstract}


\maketitle

\section{Introduction}
Graphs, representing structural data and topology, are widely used across various domains, such as social networks and merchandising transactions.
Graph convolutional networks (GCN)~\cite{iclr/KipfW17} have significantly enhanced model training on these interconnected nodes.
However, these graphs often contain sensitive information that should not be leaked to untrusted parties.
For example, companies may analyze sensitive demographic and behavioral data about users for applications ranging from targeted advertising to personalized medicine.
Given the data-centric nature and analytical power of GCN training, addressing these privacy concerns is imperative.

Secure multi-party computation (MPC)~\cite{crypto/ChaumDG87,crypto/ChenC06,eurocrypt/CiampiRSW22} is a critical tool for privacy-preserving machine learning, enabling mutually distrustful parties to collaboratively train models with privacy protection over inputs and (intermediate) computations.
While research advances (\eg,~\cite{ccs/RatheeRKCGRS20,uss/NgC21,sp21/TanKTW,uss/WatsonWP22,icml/Keller022,ccs/ABY318,folkerts2023redsec}) support secure training on convolutional neural networks (CNNs) efficiently, private GCN training with MPC over graphs remains challenging.

Graph convolutional layers in GCNs involve multiplications with a (normalized) adjacency matrix containing $\numedge$ non-zero values in a $\numnode \times \numnode$ matrix for a graph with $\numnode$ nodes and $\numedge$ edges.
The graphs are typically sparse but large.
One could use the standard Beaver-triple-based protocol to securely perform these sparse matrix multiplications by treating graph convolution as ordinary dense matrix multiplication.
However, this approach incurs $O(\numnode^2)$ communication and memory costs due to computations on irrelevant nodes.
Integrating existing cryptographic advances, the initial effort of SecGNN~\cite{tsc/WangZJ23,nips/RanXLWQW23} requires heavy communication or computational overhead.
Recently, CoGNN~\cite{ccs/ZouLSLXX24} optimizes the overhead in terms of  horizontal data partitioning, proposing a semi-honest secure framework.
Research for secure GCN over vertical data  remains nascent.

Current MPC studies, for GCN or not, have primarily targeted settings where participants own different data samples, \ie, horizontally partitioned data~\cite{ccs/ZouLSLXX24}.
MPC specialized for scenarios where parties hold different types of features~\cite{tkde/LiuKZPHYOZY24,icml/CastigliaZ0KBP23,nips/Wang0ZLWL23} is rare.
This paper studies $2$-party secure GCN training for these vertical partition cases, where one party holds private graph topology (\eg, edges) while the other owns private node features.
For instance, LinkedIn holds private social relationships between users, while banks own users' private bank statements.
Such real-world graph structures underpin the relevance of our focus.
To our knowledge, no prior work tackles secure GCN training in this context, which is crucial for cross-silo collaboration.

To realize secure GCN over vertically split data, we tailor MPC protocols for sparse graph convolution, which fundamentally involves sparse (adjacency) matrix multiplication.
Recent studies have begun exploring MPC protocols for sparse matrix multiplication (SMM).
ROOM~\cite{ccs/SchoppmannG0P19}, a seminal work on SMM, requires foreknowledge of sparsity types: whether the input matrices are row-sparse or column-sparse.
Unfortunately, GCN typically trains on graphs with arbitrary sparsity, where nodes have varying degrees and no specific sparsity constraints.
Moreover, the adjacency matrix in GCN often contains a self-loop operation represented by adding the identity matrix, which is neither row- nor column-sparse.
Araki~\etal~\cite{ccs/Araki0OPRT21} avoid this limitation in their scalable, secure graph analysis work, yet it does not cover vertical partition.

To bridge this gap, we propose a secure sparse matrix multiplication protocol, \osmm, achieving \emph{accurate, efficient, and secure GCN training over vertical data} for the first time.

\subsection{New Techniques for Sparse Matrices}
The cost of evaluating a GCN layer is dominated by SMM in the form of $\adjmat\feamat$, where $\adjmat$ is a sparse adjacency matrix of a (directed) graph $\graph$ and $\feamat$ is a dense matrix of node features.
For unrelated nodes, which often constitute a substantial portion, the element-wise products $0\cdot x$ are always zero.
Our efficient MPC design 
avoids unnecessary secure computation over unrelated nodes by focusing on computing non-zero results while concealing the sparse topology.
We achieve this~by:
1) decomposing the sparse matrix $\adjmat$ into a product of matrices (\S\ref{sec::sgc}), including permutation and binary diagonal matrices, that can \emph{faithfully} represent the original graph topology;
2) devising specialized protocols (\S\ref{sec::smm_protocol}) for efficiently multiplying the structured matrices while hiding sparsity topology.

\subsubsection{Sparse Matrix Decomposition}
We decompose adjacency matrix $\adjmat$ of $\graph$ into two bipartite graphs: one represented by sparse matrix $\adjout$, linking the out-degree nodes to edges, the other 
by sparse matrix $\adjin$,
linking edges to in-degree nodes.


We then permute the columns of $\adjout$ and the rows of $\adjin$ so that the permuted matrices $\adjout'$ and $\adjin'$ have non-zero positions with \emph{monotonically non-decreasing} row and column indices.
A permutation $\sigma$ is used to preserve the edge topology, leading to an initial decomposition of $\adjmat = \adjout'\sigma \adjin'$.
This is further refined into a sequence of \emph{linear transformations}, 
which can be efficiently computed by our MPC protocols for 
\emph{oblivious permutation}
and \emph{oblivious selection-multiplication}.
Our decomposition approach is not limited to GCNs but also general~SMM 
by 
treating them 
as adjacency matrices.


\subsubsection{New Protocols for Linear Transformations}
\emph{Oblivious permutation} (OP) is a two-party protocol taking a private permutation $\sigma$ and a private vector $\xvec$ from the two parties, respectively, and generating a secret share $\l\sigma \xvec\r$ between them.
Our OP protocol employs correlated randomnesses generated in an input-independent offline phase to mask $\sigma$ and $\xvec$ for secure computations on intermediate results, requiring only $1$ round in the online phase (\cf, $\ge 2$ in previous works~\cite{ccs/AsharovHIKNPTT22, ccs/Araki0OPRT21}).

Another crucial two-party protocol in our work is \emph{oblivious selection-multiplication} (OSM).
It takes a private bit~$s$ from a party and secret share $\l x\r$ of an arithmetic number~$x$ owned by the two parties as input and generates secret share $\l sx\r$.
Our $1$-round OSM protocol also uses pre-computed randomnesses to mask $s$ and $x$.
Compared to the Beaver-triple-based~\cite{crypto/Beaver91a} and oblivious-transfer (OT)-based approaches~\cite{pkc/Tzeng02}, our protocol saves ${\sim}50\%$ of online communication while having the same offline communication and round complexities.

By decomposing the sparse matrix into linear transformations and applying our specialized protocols, our \osmm protocol
reduces the complexity of evaluating $\numnode \times \numnode$ sparse matrices with $\numedge$ non-zero values from $O(\numnode^2)$ to $O(\numedge)$.

\subsection{\cgnn: Secure GCN made Efficient}
Supported by our new sparsity techniques, we build \cgnn, 
a two-party computation (2PC) framework for GCN inference and training over vertical
data.
Our contributions include:

1) We are the first to explore sparsity over vertically split, secret-shared data in MPC, enabling decompositions of sparse matrices with arbitrary sparsity and isolating computations that can be performed in plaintext without sacrificing privacy.

2) We propose two efficient $2$PC primitives for OP and OSM, both optimally single-round.
Combined with our sparse matrix decomposition approach, our \osmm protocol ($\prosmm$) achieves constant-round communication costs of $O(\numedge)$, reducing memory requirements and avoiding out-of-memory errors for large matrices.
In practice, it saves $99\%+$ communication
and reduces ${\sim}72\%$ memory usage over large $(5000\times5000)$ matrices compared with using Beaver triples.

3) We build an end-to-end secure GCN framework for inference and training over vertically split data, maintaining accuracy on par with plaintext computations.
We will open-source our evaluation code for research and deployment.

To evaluate the performance of $\cgnn$, we conducted extensive experiments over three standard graph datasets (Cora~\cite{aim/SenNBGGE08}, Citeseer~\cite{dl/GilesBL98}, and Pubmed~\cite{ijcnlp/DernoncourtL17}),
reporting communication, memory usage, accuracy, and running time under varying network conditions, along with an ablation study with or without \osmm.
Below, we highlight our key achievements.

\textit{Communication (\S\ref{sec::comm_compare_gcn}).}
$\cgnn$ saves communication by $50$-$80\%$.
(\cf,~CoGNN~\cite{ccs/KotiKPG24}, OblivGNN~\cite{uss/XuL0AYY24}).

\textit{Memory usage (\S\ref{sec::smmmemory}).}
\cgnn alleviates out-of-memory problems of using 
Beaver-triples~\cite{crypto/Beaver91a} for large datasets.

\textit{Accuracy (\S\ref{sec::acc_compare_gcn}).}
$\cgnn$ achieves inference and training accuracy comparable to plaintext counterparts.

{\textit{Computational efficiency (\S\ref{sec::time_net}).}} 
$\cgnn$ is faster by $6$-$45\%$ in inference and $28$-$95\%$ in training across various networks and excels in narrow-bandwidth and low-latency~ones.

{\textit{Impact of \osmm (\S\ref{sec:ablation}).}}
Our \osmm protocol shows a $10$-$42\times$ speed-up for $5000\times 5000$ matrices and saves $10$-2$1\%$ memory for ``small'' datasets and up to $90\%$+ for larger ones.

\section{Preliminary}

\paragraph{Notations.}
Table~\ref{tab:notation} summarizes the main notations.
$\Abb$ denotes an Abelian group.
$\Sbb_n$ denotes a permutation group of $n$ elements.
$\Mbb_{m,n}(\Rcal)$ denotes a matrix ring, which defines a set of $m\times n$ matrices with entries in a ring $\Rcal$, forming a ring under matrix addition and
multiplication.
$\Msf_{m \times n}=(\Msf[i,j])_{i,j=1}^{m,n}$ denotes an $m\times n$ matrix\footnote{
For simplicity, we omit the subscript of 
$\Msf_{m \times n}$ when the values of $m$ and $n$ are clear from the context.
Also, we write 
$\Msf = (\Msf[i, j])_{i, j = 1}^{n}$ if $m = n$.
} 
where row indices are $\{1, 2, \ldots, m\}$ and column indices are $\{1, 2, \ldots, n\}$, and $\Msf[i,j]$ is the value at the $i$-th row and $j$-th column.
$\Pi(\ ;\ )$ denotes a protocol execution between two parties, $\pp_0$ and $\pp_1$, where $\pp_0$'s inputs are the left part of `;' and $\pp_1$'s inputs are the right part of `;'.

\paragraph{Secret Sharing.} 
We use $2$-out-of-$2$ additive secret sharing over a ring, where
the floating-point
values are encoded to be fixed-point numbers $x\in \Zbb/2^f\Zbb$, with $L = 64$ bits representing decimals and $f = 18$ bits representing the fractional part~\cite{sp/MohasselZ17}.
Specifically, one party $\pp_0$ holds the share $\l x\r _0\in\Zbb $, while the other party $\pp_1$ holds the share $\l x\r_1\in\Zbb$ such that $x\cdot 2^f = \l x\r_0 + \l x\r_1$.
The shares can be arithmetic or binary.


\paragraph{Graph Convolutional Networks.}
GCN~\cite{iclr/KipfW17} has been proposed for training over graph data, using graph structure and node features as input.
Like most neural networks, GCN consists of multiple linear and non-linear layers.
Compared to CNN, GCN replaces convolutional layers with graph convolution layers
(more details on GCN architecture and its training/inference are in Section~\ref{sec::secgcn}).
Graph convolution can be computed by SMM, often yielding many $0$-value results.

Let $\adjmat \in \Mbb_{\numnode,\numnode}(\Rcal)$ be a (normalized) adjacency matrix of a graph with $\numnode$ nodes and $\Xsf \in \Mbb_{\numnode, d}(\Rcal)$ be the feature matrix (with dimensionality $d$) of the nodes.
The graph convolution layer (with output dimensionality $k$) is defined as $\mathsf{Y} = \adjmat \feamat \weimat$,
where $\mathsf{Y} \in \Mbb_{\numnode, k}(\Rcal)$ is the output and $\weimat \in \Mbb_{d, k}(\Rcal)$ is a trainable parameter.
As matrix multiplication costs increase linearly with input size and $k \ll \numnode$ in practice, the challenge of secure GCN
lies in the
SMM of $\adjmat \feamat$.
Multiplying (dense) $\weimat$ can be done using Beaver's approach.

\begin{table}[!t]
\centering
\caption{Notation and Definition}
\label{tab:notation}
\setlength\tabcolsep{2pt}
\begin{tabular}{l|l}
\hline
$\pp_i, \l \ \r_i$& Party $i$ and its share ($i \in \{0, 1\})$
\\\hline
$\bitlen$ & The bit-length of data
\\\hline
$\pi,b,\uvec,\l u \r$ & Pre-computed randomnesses
\\\hline
$\delta_x$ & Masked version of value/vector $x$
\\\hline
$\Mbb_{m, n}{(\Rcal)}$ & A set of ${m\times n}$ matrices with entries in a ring $\Rcal$
\\\hline
$\Msf_{m,n}$ & A matrix $\Msf$ of size $m\times n$
\\\hline
$\Msf[i, j]$ & The value of $\Msf$ at the $i$-th row and $j$-th column
\\\hline
$\sigma\xvec$& Permutation operation $\sigma$ over a matrix/vector $\xvec$
\\\hline
$\Sbb_n$ & Permutation group of $n$ elements
\\\hline
\end{tabular}
\end{table}

\section{System Overview and Security Model}
\label{sec::securitymodel}
\subsection{Workflow of \texorpdfstring{$\cgnn$}{\textcgnn}}
Figure~\ref{graph_imple} outlines $\cgnn$'s function.
A graph owner $\pp_0$, 
with an adjacency matrix $\adjmat$ corresponding to a private graph $\graph$,
and a feature owner $\pp_1$ with private node features $\feamat$, aim to jointly train a GCN without revealing their private inputs.
This involves computing a parameterized function $\mathsf{GCN}(\adjmat, \feamat; \weimat)$, where the weights $\weimat$ are secret-shared over the two parties.

The $\cgnn$ framework includes a sparse matrix decomposition method (Section~\ref{sec::sgc}) and secure $2$PC protocols for permutation ($\Pi_{\ssp}$, Section~\ref{sec::op_pro}),
selection-multiplication ($\Pi_{\SM}$, Section~\ref{sec::osm_pro}), and SMM ($\prosmm$, Section~\ref{subsec:prosmm}).
The sparse matrix decomposition is performed solely by the graph owner, while all $2$PC protocols are executed by both parties without disclosing any intermediate computation results.

In practical cross-institution collaboration, graph owners can be social networking platforms (\eg, Facebook) holding social relationships as a graph, and feature owners can be banks holding users' bank statements as node features.
As a motivating example, they may want to build a credit-investigation model for predicting the credit of a loaner for future repayment
while keeping their data confidential.
Our setting can be extended to multi-party, where different types of node features are learned from different parties (\eg, bank statements from banks and transactions from online-shopping companies).
Usually, the graph structure is  fixed to represent a specific relationship, such as a social circle, in real-world scenarios.
Thus, we focus on single-party graph ownership without limiting feature ownerships.
A general case of arbitrary partitioning has been discussed in Section~\ref{sec:future}.
 
\begin{figure}[!t]
	\centering
	\includegraphics[width = 0.42\textwidth]{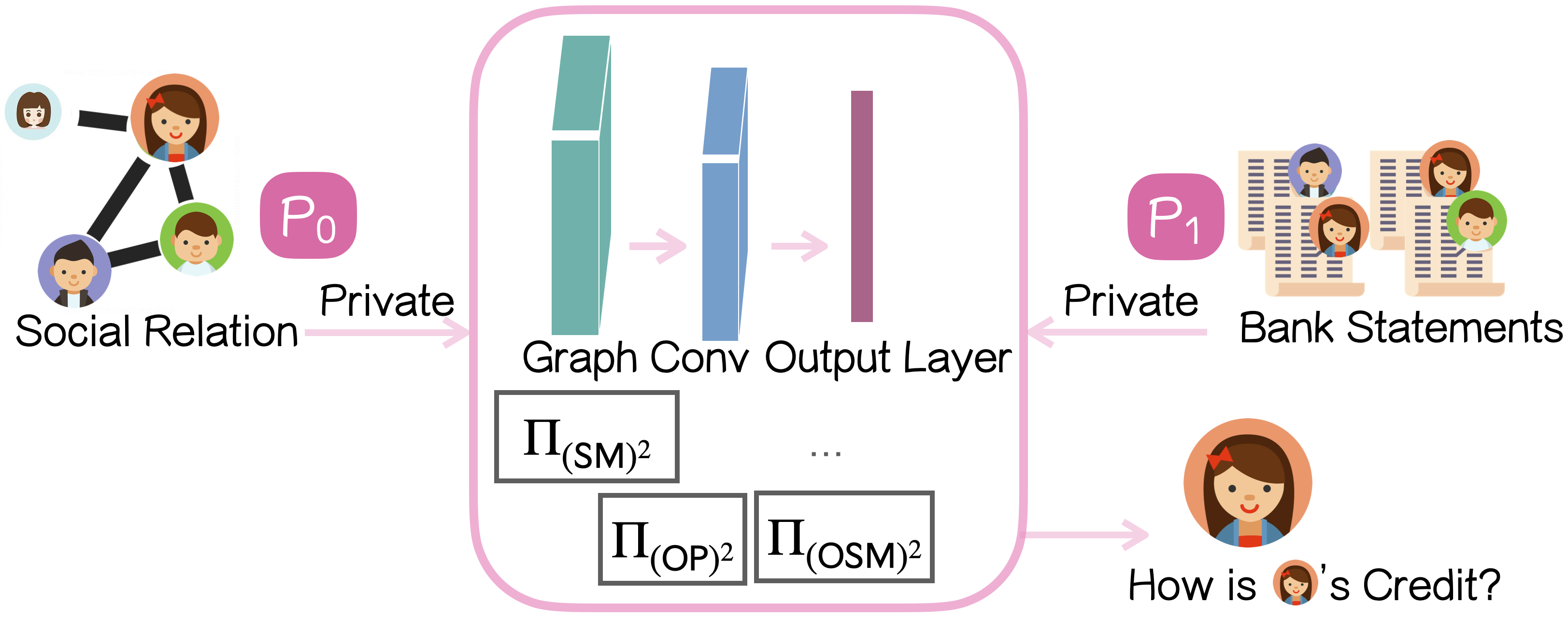}
	\caption{Ideal Functionality of $\cgnn$}
	\label{graph_imple}
\end{figure}

\subsection{Security Model}
$\cgnn$ can be instantiated with any type of security models offered by the corresponding MPC protocols.
Following advances~\cite{ccs/AttrapadungH0MO21,neurips/crypten2020,sp21/TanKTW,uss/WatsonWP22,ccs/ZouLSLXX24,ccs/KotiKPG24}, $\cgnn$ focuses on $2$PC security against the static semi-honest probabilistic polynomial time (PPT) adversary $\A$ regarding the real/ideal-world simulation paradigm~\cite{sp/17/Lindell17}.
Specifically, two parties, $\pp_0$ and $\pp_1$, with inputs $\l x\r_0$ and $\l x\r_1$, want to compute a function $y = f(\l x\r_0, \l x\r_1)$ without revealing anything other than $y$.
$\A$ corrupts either $\pp_0$ or $\pp_1$ at the start, following the protocol, but tries to learn the other's private inputs.
$\A$ can only learn data from the corrupted party but nothing from  honest ones.

Many protocols utilize pre-computations for improving efficiency, \eg, Beaver triples~\cite{crypto/Beaver91a} for multiplication.
They can be realized by a data-independent offline phase run by a  semi-honest dealer $\tt$ or 2PC protocols from homomorphic encryption~\cite{eurocrypt/LyubashevskyPR10} or oblivious transfer~\cite{pkc/Tzeng02,ccs/KellerOS16} or oblivious shuffle~\cite{asiacrypt/ChaseGP20,ndss/SongYBDC23}.
We adopt the first common approach (also called client-aided setting~\cite{ccs/AttrapadungH0MO21}) for simplicity.
The $\tt$ does not interact with any party (particularly, receives nothing) online.
It only generates pseudo-randomnesses in an input-independent offline phase by counter-indexed computations of pseudorandom function (PRF), where $\tt$ and $\pp_i$ share a PRF key (denoted by $\key_i$) for $i \in \{0,1\}$ and a counter $\ctr$ are synchronized among all parties.
We defer the explicit functionality definitions and security proofs of our protocols to Appendix~\ref{app::fullproof}.

\subsection{Scope of Graph Protection}

Like existing MPC works, $\cgnn$ protects the entry values stored in the graph and (intermediate) computations.
For metadata, most secure matrix multiplication protocols (without sparse structure) reveal input
dimensionality (\eg, $\numnode$ in GCN) that is typically considered public knowledge.
When sparsity is explored, it is normal to leak reasonable knowledge, such as $\numnode+\numedge$ in GraphSC~\cite{sp/NayakWIWTS15}.
In $\cgnn$, the only additional metadata revealed is $\numedge$ that has been comprehensively explained in \S\ref{subsec:prosmm}.
This leakage is tolerable (and unavoidable) since the efficiency gain is correlated to $\numedge$.
Corresponding to $\cgnn$'s GCN training, the \textit{dimension} of adjacency matrix $\adjmat$ (\ie, equal to $\numnode$) and the \textit{dimension} of feature matrix $\Xsf$ are assumed to be public. 

Privacy leakages from training/inference results, \eg, embedding inversion and sensitive attribute inference, also appear in plaintext computations and are beyond our scope.
These can be protected via orthogonal techniques like (local) differential privacy and robustness training, which are compatible with our work.
In the semi-honest settings, the attacker can only view the well-formed secret shares and not actively perform the malicious attacks like model inversion.
\section{Sparse Matrix Decomposition}
\label{sec::sgc}
Graph convolution layers $\adjmat\feamat\weimat$ encode the graph structures in $\adjmat$ into GCNs.
Graph convolution is then computed by SMM $\adjmat\feamat$.
By bridging computations of matrices and graphs, we detail how to decompose a sparse matrix $\adjmat$ into a product of special matrices for more efficient SMM.
In essence, we \emph{revisit linear algebra} relations to \emph{faithfully} capture the graph.

\subsection{Bipartite Graph Representation}\label{subsec::sme}
We represent graph $\graph$ corresponding to $\adjmat$ as bipartite graphs, and decompose $\adjmat$ into matrices.
This bipartite representation enables the identification of structured patterns that facilitate efficient SMM aligning with our $2$PC protocols.

\paragraph{Graph Decomposition via Edges.}
Non-zero entries in $\adjmat$ correspond to edges between nodes in $\graph$.
By representing $\graph$ as two bipartite graphs---$\graphout$ (the out-degree node-to-edge relation) and $\graphin$ (the in-degree edge-to-node relation)---we can decompose $\adjmat$ into the product $\adjout \adjin$, where $\adjout$ and $\adjin$ reflect the respective bipartite structures 
are sparse matrices correspond to $\graphout$ and $\graphin$, respectively.
Consider graph~$\graph$ (with arbitrary-sparse $\adjmat$) in Figure~\ref{fig::nn_relation_diff}.
We label each edge and treat them as imaginary nodes (`$\diamond$' drawn by dotted lines) to
construct $\graphout$ and $\graphin$ as in Figure~\ref{fig::nen_relation_diff}.
This representation decomposes $\adjmat$ into $\adjout \adjin$ as in Figure~\ref{fig::adjmat_decom_init}.


\begin{figure}[!t]
 \subfloat[Node-Node Graph $\graph$]{
 \includegraphics[width = 0.21\textwidth]{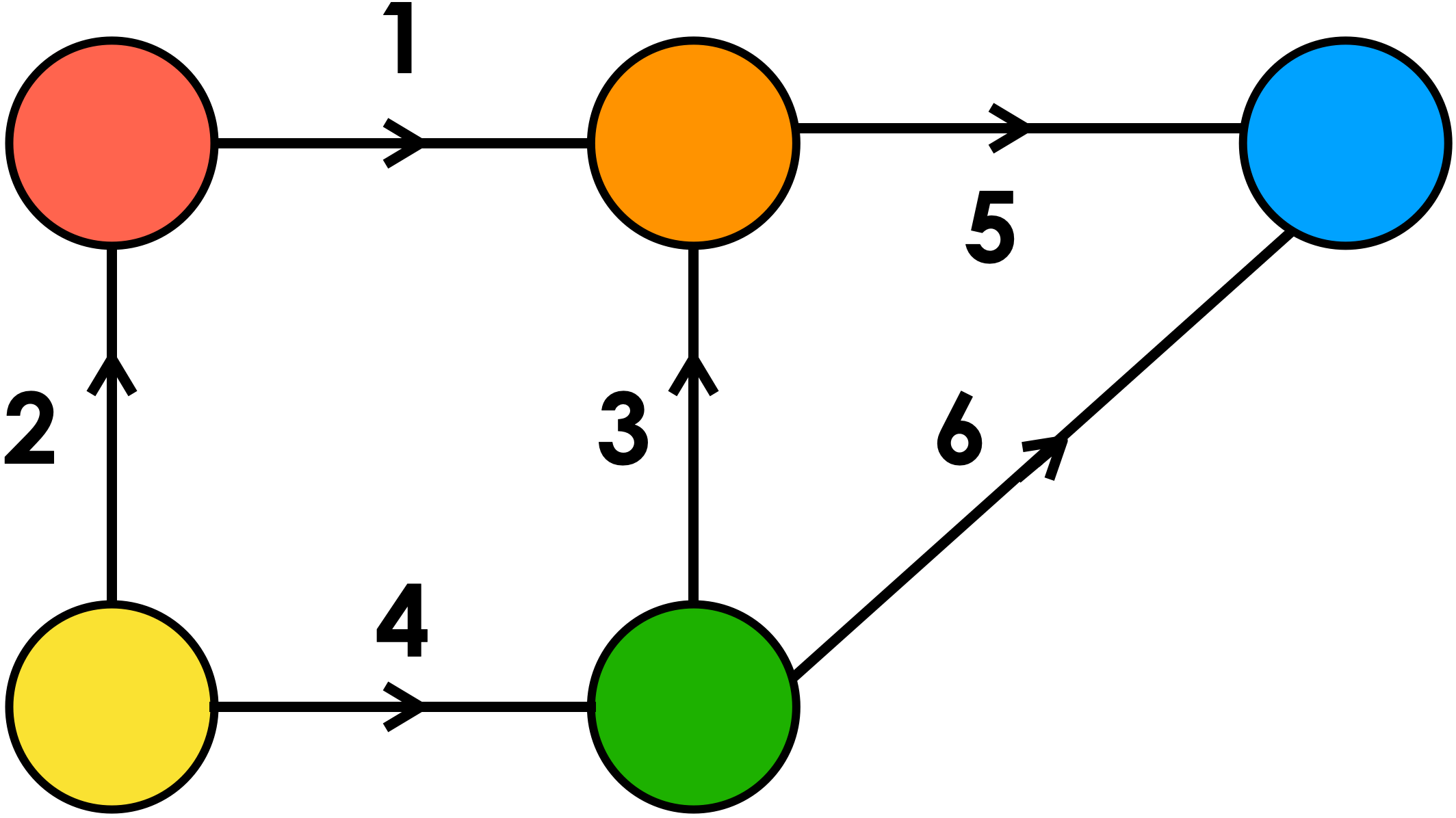}
 \label{fig::nn_relation_diff}
 }
 \hspace{2mm}
 \subfloat[Node-Edge-Node Graph $\graphout, \graphin$]{
 \includegraphics[width = 0.225\textwidth]{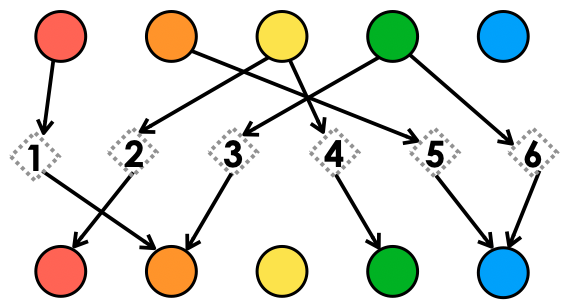}
 \label{fig::nen_relation_diff}
 }
	\caption{Graph Decomposition through Edges}
	\label{fig::graph_decom_diff}
\end{figure}

\begin{figure}[!t]

\centering
	\includegraphics[width = 0.47\textwidth]{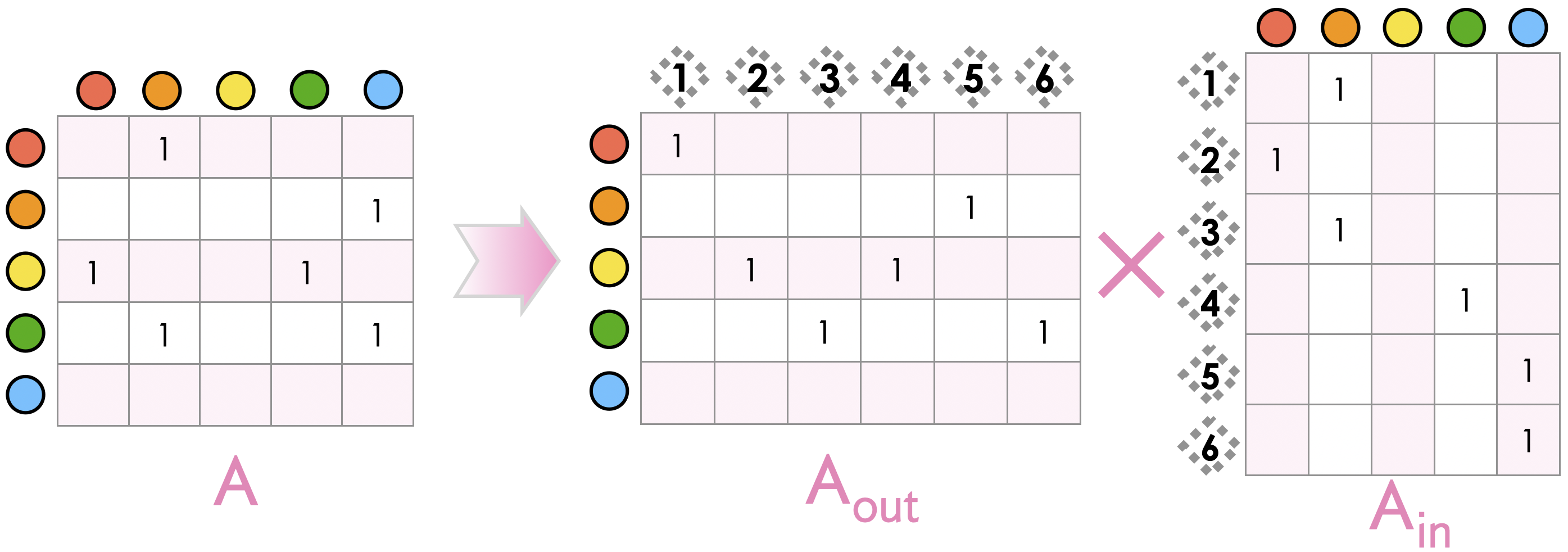}
	\caption{Matrix Decomposition Equivalent to Figure~\ref{fig::graph_decom_diff}}
	\label{fig::adjmat_decom_init}
\end{figure}
 
\begin{figure*}[!t]
	\centering
	\includegraphics[width = 0.98\textwidth]{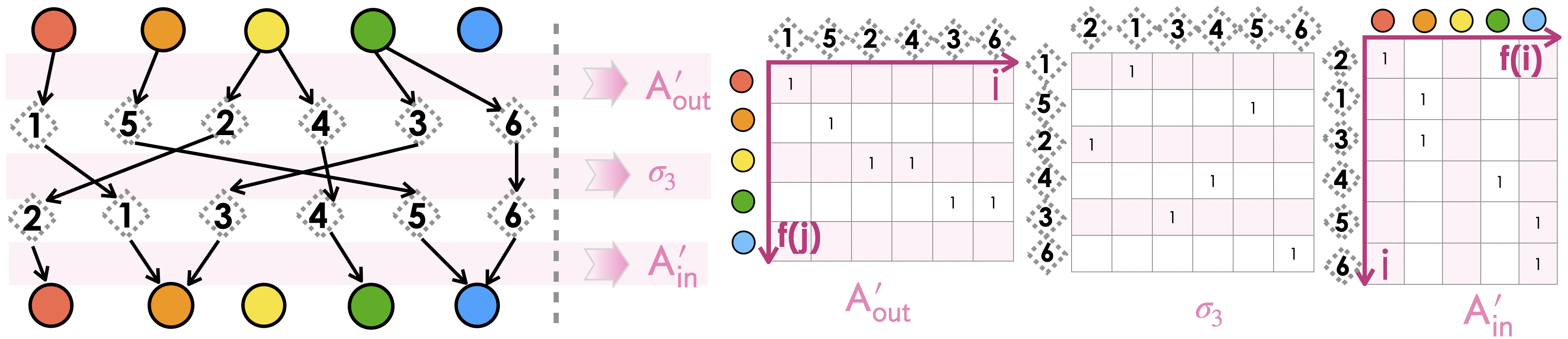}
	\caption{Graph/Matrix Decomposition with Monotonicity}
	\label{graph_psigq}
\end{figure*}

\subsection{Permutation for Monotonicity}\label{subsec::initdecom}
$\adjout$ and $\adjin$ are still unstructured sparse matrices, challenging further decomposition.
We then permute the columns of $\adjout$ and the rows of $\adjin$ to yield permuted matrices $\adjout'$ and $\adjin'$ with monotonically non-decreasing (row-index, column-index) coordinates for non-zero positions as shown in Figure~\ref{graph_psigq}.
Definitions~\ref{q_type} and~\ref{p_type} formulate $\Psf$-type and $\Qsf$-type
sparse matrices to capture these monotonic relations,
where $\Psf$-type matrices have exactly one non-zero value in each column, and
$\Qsf$-type matrices have exactly one non-zero value in each row.
Note that $\adjmat \neq \adjout'\adjin'$ as the imaginary nodes
in $\adjout'$ and $\adjin'$ are ordered differently.
We use a permutation $\sigma_3$ 
(before defining $\sigma_1,\sigma_2$)
to map these nodes
for preserving the topology among edges and 
decomposing 
$\adjmat$, given by $\adjout'\sigma_3\adjin'$.

Recall that $\adjmat$ is a normalized adjacency matrix, \ie, its non-zero values may not be $1$.
To account for this, we introduce a diagonal matrix $\Lambda$ in the decomposition to store the non-zero edge weights.
Theorem~\ref{the::p_sig_q} (proven in $\S$\ref{sec::proof_p_sig_q})
shows that any sparse matrix $\adjmat$ can be decomposed to a $\Psf$-type matrix, a diagonal matrix,
a permutation, and a $\Qsf$-type matrix.



\begin{restatable}{theorem}{initdecom}
\label{the::p_sig_q}
Let $\adjmat\in \Mbb_{m,n}(\Rcal)$ be an $m\times n$ matrix, where each entry is an element from ring $\Rcal$.
The elements of $\adjmat$ are $0$'s except $\numnonzero$ of them.
There exists a matrix decomposition $\adjmat = \adjout' \Lambda \sigma_3 \adjin'$, where $\adjout' \in \Mbb_{m, \numnonzero}(\Rcal)$ is a $\Psf$-type matrix, $\Lambda \in \Mbb_{\numnonzero,\numnonzero}(\Rcal)$ is a diagonal matrix, $\sigma_3 \in \Sbb_{\numnonzero}$ is a permutation,
and $\adjin' \in \Mbb_{\numnonzero, n}(\Rcal)$ is a $\Qsf$-type matrix.
\end{restatable}

\subsection{Re-decomposition to Basic Operations}
\label{sec::the_supp_smm}
Given the permuted matrices $\adjout'$ and $\adjin'$ with the 
monotonicity properties, we can re-decompose them into a product of permutation, diagonal, and constant matrices.
Due to the page limit, we focus on the intuition of re-decomposing $\adjin'$ (Q-type matrix)\footnote{
We can view a P-type matrix (\eg, $\adjout'$) as a transpose of a Q-type matrix and perform re-decomposition similarly.
} and the general theorem.
Implementation details and proofs can be found in Appendices~\ref{sec:matrix_found_sparse} and~\ref{sec::alg}.

We consider two constant lower triangular matrices:
\\
1) a ``summation matrix'' $\Sigma \in \Mbb_{\numnonzero,\numnonzero}(\Rcal)$ with $\Sigma[i, j]=1$ if $i \geq j$ or $0$ otherwise;
\\
2) a ``difference matrix'' $\delta_k \in \Mbb_{k,k}(\Rcal)$ with $\delta_k[i,j]=1$ for $i=j$ or $-1$ for $j=i-1$, or $0$ otherwise.

Intuitively, when multiplying with a (column) vector, $\Sigma$ sums values on or above each element, while $\delta_k$ computes a difference between each element and its previous one.

Based on the above intuition, it is not hard to decompose $\adjin'$ into a product of $\Sigma$ and another matrix $\delta'$ (Figure~\ref{fig::Q-decom_1}).
Interestingly,
we observe that the resulting matrix $\delta'$ ``contains'' a difference matrix (with size equals the number of non-zero columns in $\adjin'$) on its left-top corner (after permuting its rows and columns).
This relation can be characterized by expressing $\delta'$ into a product of permutation ($\sigma_1$, $\sigma_2$), diagonal ($\gammain$), and difference ($\delta$) matrices, as in Figure~\ref{fig::q-type_decom}.
 
\begin{figure}[!t]
	\centering
	\includegraphics[width = 0.47\textwidth]{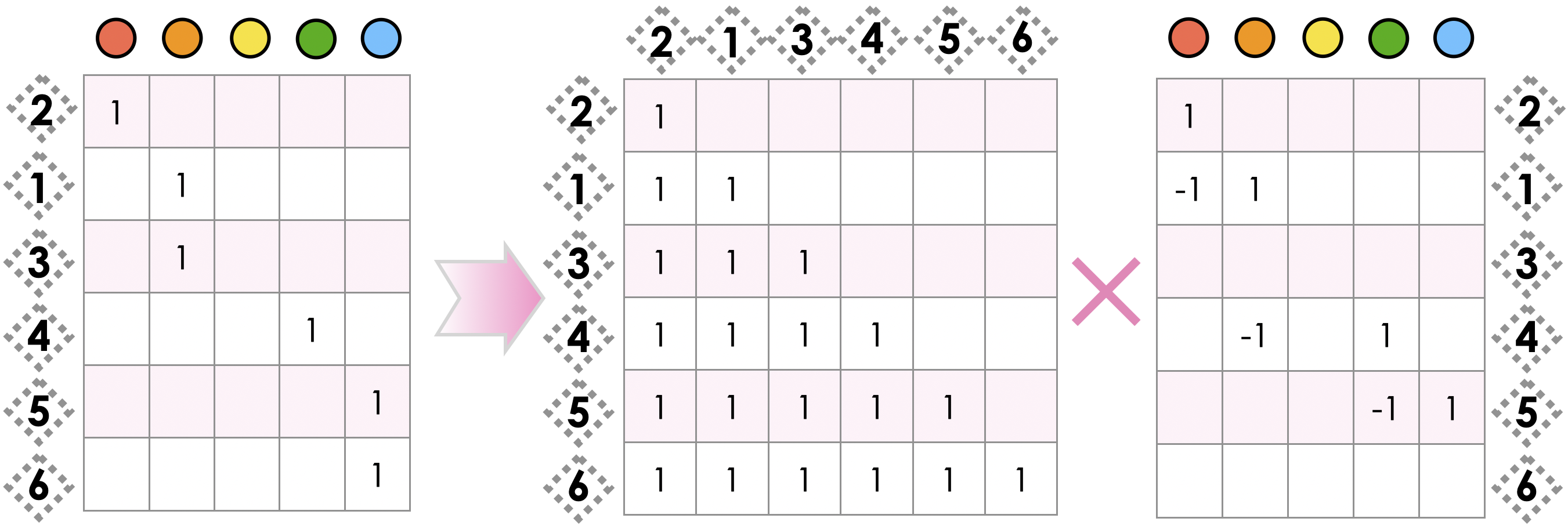}\caption{Decomposition of $\adjin' = \Sigma \delta'$}
	\label{fig::Q-decom_1}
\end{figure}

\begin{figure*}[!t]
\centering
\includegraphics[width = 0.99\textwidth]{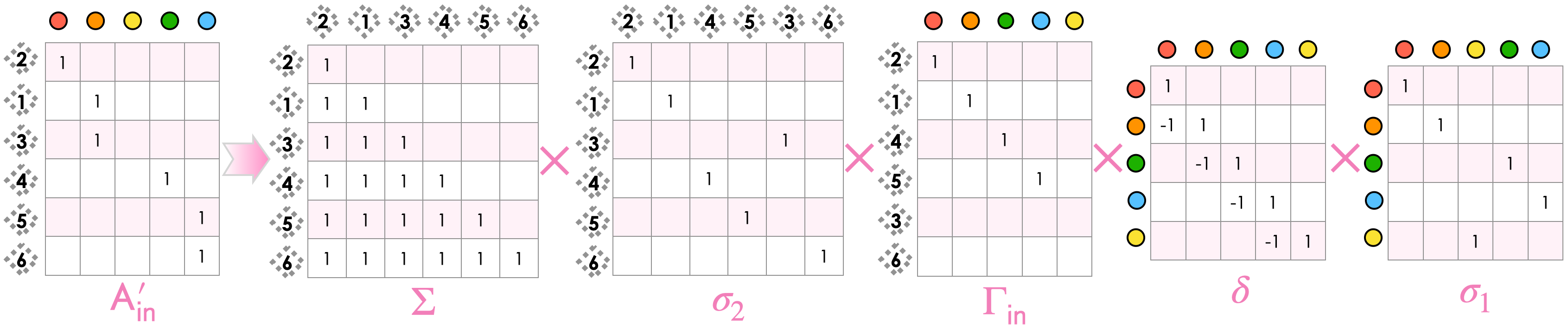}	\caption{Re-decomposition of $\adjin'$ (Q-type Matrix)}
	\label{fig::q-type_decom}
\end{figure*}

\paragraph{General Theorem.}
Combining Theorem~\ref{the::p_sig_q} and matrix decomposition of $\Qsf$-type (Theorem~\ref{the::q_decom} proved in $\S$\ref{sec::proof_q}) and $\Psf$-type (Theorem~\ref{the::p_decom} proved in $\S$\ref{sec::proof_p}) matrices, Theorem~\ref{the::general_dec_main} concludes the general matrix decomposition (proof in $\S$\ref{sec::proof_general}).
Essentially, an arbitrary-sparse matrix can be transformed into a 
sequence of permutation and matrix multiplication.

\begin{restatable}[Sparse Matrix Decomposition]{theorem}{finaldecom}
\label{the::general_dec_main}
Let an $m\times n$ sparse matrix $\adjmat\in \Mbb_{m,n}(\Rcal)$ contain $\nrow$ non-zero rows, $\ncol$ non-zero columns, and $\numnonzero$ non-zero elements.
Then, there exists a matrix decomposition
$\adjmat = \sigma_5 \delta_m ^{\top} \gammaout \sigma_4 \Sigma ^{\top} \Lambda \sigma_3 \Sigma \sigma_2 \gammain \delta _n \sigma_1$,
where $\sigma _5 \in \Sbb_m$, $\sigma_4 \in \Sbb_{\numnonzero}, \sigma_3 \in \Sbb_{\numnonzero}, \sigma_2 \in \Sbb_{\numnonzero}, \sigma_1 \in \Sbb_n$, and,
\\
1) $\Sigma=(\Sigma[i, j])_{i,j=1}^{\numnonzero}$ is the left-down triangle matrix such that $\Sigma[i, j]=1$ if $i \geq j$ or $0$ otherwise,
\\
2) $\delta_k=(\delta_k[i,j])_{i,j=1}^{k}$ is the left-down triangle matrix such that $\delta_k[i,j]=1$ for $i=j$ or $-1$ for $j=i-1$, or $0$ otherwise,
\\
3) $\gammain =(\gammain[i,j])_{i=1,j=1}^{\numnonzero,n}$ is a matrix such that $\gammain[i,j]=1$ for $1\leq i=j\leq \ncol$ or $0$ otherwise,
\\
4) $\gammaout =(\gammaout[i,j])_{i=1,j=1}^{m,\numnonzero}$ is a matrix such that $\gammaout[i,j]=1$ for $1\leq i=j\leq {\nrow}$ or $0$ otherwise.
\end{restatable}

\subsection{Reasoning from Graph Perspective}
\begin{figure}[!t]
	\centering
	\includegraphics[width = 0.477\textwidth]{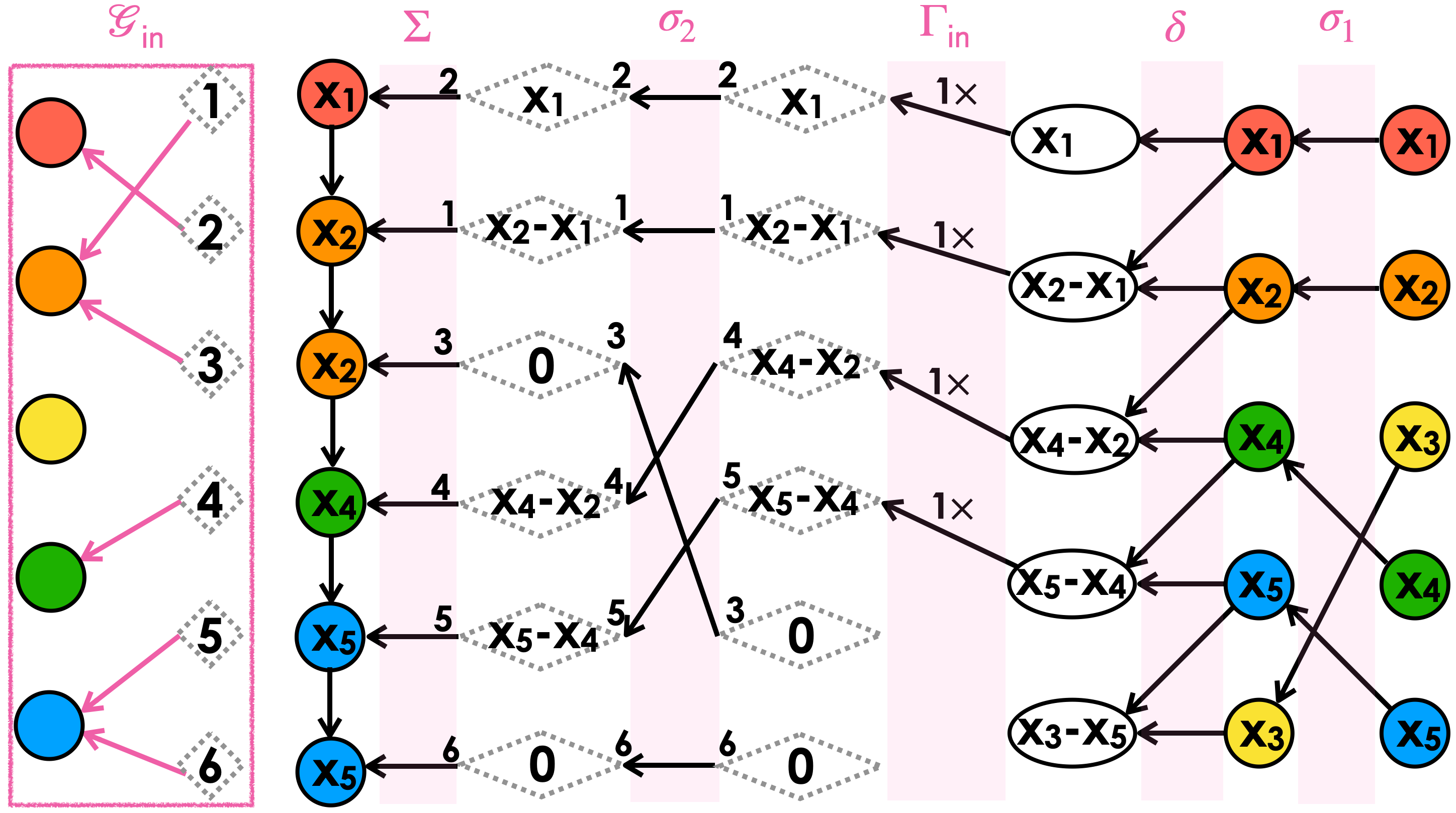}
	\caption{Recover In-degrees in $\graphin$ through $\adjin'\Xsf$}
	\label{graph_psigqx_diff}
\end{figure}

To illustrate the sparse matrix decomposition underlying Theorem~\ref{the::general_dec_main} for arbitrary topology,
Figure~\ref{graph_psigqx_diff} shows the directed edges \textit{in a reversed direction} represented by the decomposed matrices in Figure~\ref{fig::q-type_decom} and recovers the original $\graphin$.
Consider a vector
$\Xsf = [x_1, x_2, x_3, x_4, x_5]^\top$, 
which passes $5$ values through graph $\graphin$ (equivalent to SMM $\adjin'\Xsf$).
After $\sigma_1$ operation, 
$\Xsf$ passes from $\{\tikzcircle[fill=red]{3pt},\tikzcircle[fill=orange]{3pt},\tikzcircle[fill=yellow]{3pt},\tikzcircle[fill=ao]{3pt},\tikzcircle[fill=cyan]{3pt}\}$ to the re-ordered $\{\tikzcircle[fill=red]{3pt},\tikzcircle[fill=orange]{3pt},\tikzcircle[fill=ao]{3pt},\tikzcircle[fill=cyan]{3pt},\tikzcircle[fill=yellow]{3pt}\}$.
Then, the $\delta$ operation computes the difference of values stored in the neighboring source nodes to obtain the target nodes.

The $\gammain$ operation extracts the effective message passing to the subsequent graph computation by classifying the nodes with or without in-degree edges in $\graphin$.
Thus, all interdependent nodes
$\{x_1 (\tikzcircle[fill=red]{3pt}), 
x_2 (\tikzcircle[fill=orange]{3pt}), 
x_4 (\tikzcircle[fill=ao]{3pt}), 
x_5 (\tikzcircle[fill=cyan]{3pt})\}$ 
in $\graphin$ are recovered, \ie, those nodes containing one or multiple in-degree edges.

After the $\sigma_2$ operation, nodes
are rearranged in order
$\{x_1 (\tikzcircle[fill=red]{3pt}),\allowbreak 
x_2 (\tikzcircle[fill=orange]{3pt}), 
x_3 (\mathrm{None}), 
x_4 (\tikzcircle[fill=ao]{3pt}), 
x_5 (\tikzcircle[fill=cyan]{3pt})\}$.
Interestingly, the imaginary nodes (`$\diamond$' drawn by dotted lines) reflect the $\delta'$ matrix in Figure~\ref{fig::Q-decom_1}.
Next, the $\Sigma$ operation takes the sum of source nodes to get target nodes.
Finally, we get the results
$\{\tikzcircle[fill=red]{3pt},
\tikzcircle[fill=orange]{3pt},
\tikzcircle[fill=orange]{3pt},
\tikzcircle[fill=ao]{3pt},
\tikzcircle[fill=cyan]{3pt},
\tikzcircle[fill=cyan]{3pt}\}$, 
which recover the (permuted) in-degree edges (represented by $\adjin'$) matching the correct nodes in~$\graphin$.

\section{Secure Sparse Matrix Multiplication}
\label{sec::smm_protocol}
Given the sparse matrix decomposition from Theorem~\ref{the::general_dec_main}, SMM can be 
transformed into an ordered sequence of basic operations from right to left as Theorem~\ref{the::smm_main} (proof in $\S$\ref{sec::proof_smm}).
If we expect to compute $\feamat\adjmat$, the linear transformations should be performed sequentially from left to right.
For a sparse matrix that is multiplied by another sparse matrix, we can combine the sequential computation of $\adjmat\feamat$ and $\feamat\adjmat$.

\begin{restatable}[Sparse Matrix Multiplication]{theorem}{thmsmm}
\label{the::smm_main}
Consider a sparse matrix $\adjmat$ and a dense matrix $\feamat$.
Computing $\adjmat\feamat =
\allowbreak \sigma_5 \delta_m ^{\top} \gammaout \sigma_4 \Sigma ^{\top} \Lambda \sigma_3\allowbreak \Sigma \sigma_2 \gammain \delta_n \sigma_1 \feamat$ requires an ordered sequence of permutation group action, element-wise multiplication, 
and constant matrix multiplication from right to left.
\end{restatable}

For secure MPC, the graph owner $\pp_0$ first decomposes its graph to obtain matrices $\sigma_1, \sigma_2, \sigma_3, \sigma_4, \sigma_5, \gammaout, \gammain, \Lambda$.
These matrices are privacy-sensitive and should not be learned by the feature owner $\pp_1$.
The summation matrix $\Sigma$ and difference matrices $\delta_m, \delta_n$ are constants (given dimensionality of~$\adjmat$) and thus are public to both $\pp_0$ and $\pp_1$.
Next, $\pp_0$ and $\pp_1$ jointly execute the MPC protocols of SMM, which multiplies the above matrices described in Theorem~\ref{the::smm_main}.

We first present an oblivious permutation protocol (for secure permutation operations based on $\sigma_1, \ldots, \sigma_5$) in Section~\ref{sec::op_pro}
and then an oblivious selection-multiplication protocol (for privately multiplying $\gammaout$ and $\gammain$) in Section~\ref{sec::osm_pro}.
Finally, we describe how to realize our \osmm protocol using our OP and OSM protocols in Section~\ref{subsec:prosmm}.

\subsection{Oblivious Permutation}
\label{sec::op_pro}
Protocol~$\Pi_{\ssp}$
is our oblivious permutation protocol.
Given $\pp_0$'s private permutation $\sigma\in\Sbb_k$ and $\pp_1$'s private $k$-dimensional vector $\xvec\in\Zbb^k_{2^L}$, $\Pi_{\ssp}$ generates a secret share $\l \sigma\xvec \r_i$ for $\pp_i \in \{\pp_0, \pp_1\}$ without revealing $\sigma$ or $\xvec$.
The protocol parameter $\type \in \{\raw, \shared\}$ specifies the type of input vector $\xvec$.
If $\type$ is $\raw$, $\xvec$ is initially owned by $\pp_1$; otherwise, it is secret-shared among $\pp_0$ and $\pp_1$.

\begin{protocol}[!t]
\caption{$\Pi_{\ssp}$: Oblivious Permutation}\label{fig::sspcons:main}
\begin{algorithmic}[1]
	\item[\textbf{Parameter:} $\pp_0$ and $\pp_1$ know $\type \in \{\raw, \shared\}$.]
	\REQUIRE $\pp_0$ inputs $\sigma$ and $\pp_1$ inputs $\xvec$ if $\type == \raw$;
	\\~~~~~otherwise, $\pp_0$ inputs $(\sigma, \l\xvec\r_0)$ and $\pp_1$ inputs $\l\xvec\r_1$.
	\ENSURE $\pp_0$ gets $\l \sigma\xvec \r_0$ and $\pp_1$ gets $\l \sigma\xvec \r_1$.
	\STATE {\color{gray}\COMMENT{Offline Phase: Generate Correlated Randomness}}
	\STATE $\tt, \pp_0$: Get $\pi, \l \pi\uvec \r_0 \leftarrow \prf(\key_0, \ctr)$
	\STATE $\tt, \pp_1$: Get $\uvec \leftarrow \prf(\key_1, \ctr)$
	\STATE $\tt$: Send $\l \pi\uvec \r_1 = \pi\uvec - \l \pi\uvec \r_0$ to $\pp_1$
	\STATE {\color{gray}\COMMENT{Online Phase: Compute $\l \sigma\xvec \r$ in 1 Round}}
	\STATE $\pp_0$: Send $\delta_{\sigma} = \sigma \cdot \pi^{-1}$ to $\pp_1$\label{pro:op:masksigma}
	\IF{ $\type == \raw $}
		\STATE $\pp_1$: Send $\delta_{\xvec} = \xvec-\uvec$ to $\pp_0$\label{pro:op:maskx}
		\STATE $\pp_0$: Compute $\l \sigma\xvec\r_0={\color{colora}\sigma\delta_{\xvec} + \delta_{\sigma} \l \pi\uvec \r_0}$\label{pro:op:reconst0raw}
	\ELSE
		\STATE $\pp_1$: Send $\delta_{\l \xvec \r_1} = \l \xvec \r_1-\uvec$ to $\pp_0$\label{pro:op:sendmaskedx1}
		\STATE $\pp_0$: Compute $\l \sigma\xvec\r_0 = {\color{colora} \sigma\delta_{\l \xvec\r_1} + \delta_{\sigma} \l \pi\uvec \r_0 + \sigma \l \xvec \r_0}$\label{pro:op:reconst0shared}
	\ENDIF
	\STATE $\pp_1$: Compute $\l \sigma\xvec\r_1= {\color{colorb}\delta_{\sigma} \l \pi\uvec \r_1}$\label{pro:op:reconst1}
	\RETURN $\l \sigma\xvec \r$
\end{algorithmic}
\end{protocol}
\begin{table}[!t]
\centering
		\caption{Communication for Oblivious Permutation}
				\setlength\tabcolsep{6pt}
			\begin{tabular}{l|c|c|c}
			 \hline
\textbf{Protocol} & \textbf{Offline} & \textbf{Online} & \textbf{Round} \\
\hline
Asharov~\etal~\cite{ccs/AsharovHIKNPTT22} &$0$&$6k\bitlen$ &$3$\\\hline
OLGA~\cite{ccs/AttrapadungH0MO21} & $2k(\bitlen + 32)$ &$2k\bitlen $ & $1$\\\hline
Araki~\etal~\cite{ccs/Araki0OPRT21} & $0$ & $4k\bitlen$& 2 \\ \hline 
\rowcolor{grayL}$\Pi_\ssp$ & $k\bitlen$& $ k\bitlen + k\log k$ & $1$
\\\hline
			\end{tabular}\\
$\bitlen$: bit-length of data,
$k$: degree of the permutation group 
\label{table:op_comm}
\end{table} 

\paragraph{Offline Phase.}
The commodity server $\tt$ assists $\pp_0$ and $\pp_1$ to generate a random permutation $\pi\in\Sbb_k$, a random vector $\uvec\in\Zbb^k_{2^L}$, and correlated randomnesses $\l \pi \uvec \r_0\in\Zbb^k_{2^L}, \l \pi \uvec \r_1\in\Zbb^k_{2^L}$.

\paragraph{Online Phase.}
$\pp_0$ masks $\sigma$ using random $\pi^{-1}$ (\ie, inverse permutation of $\pi$) to get random {$\delta_\sigma$ ({Line~\ref{pro:op:masksigma}})}.
If $\type$ is $\raw$, $\pp_1$ masks $\xvec$ using random $\uvec$ to get $\delta_\xvec$ (Line~\ref{pro:op:maskx}).
If $\type$ is $\shared$, 
$\pp_1$ needs not mask $\xvec$ since $\l \xvec \r_0$ is kept by $\pp_0$ 
as a part of computing $\l \sigma\xvec\r_0$ (Line~\ref{pro:op:reconst0shared}).
In this case, $\pp_1$ masks $\l \xvec\r_1$ using random $\uvec$ to get random {$\delta_{\l\xvec\r_1}$ (Line~\ref{pro:op:sendmaskedx1})}.
$\pp_0$ and $\pp_1$ can then obtain the respective secret shares $\l \sigma\xvec\r_0, \l \sigma\xvec\r_1$.

\paragraph{Correctness.}
Here, we verify that 
${\color{colora}\l\sigma\xvec\r_0} + {\color{colorb}\l\sigma\xvec\r_1} = \sigma\xvec$.
If $\type$ is $\raw$, it holds that $\sigma\Xsf = \sigma (\Xsf-\uvec + \uvec) = \sigma(\delta_{\Xsf} + \uvec) = \sigma\delta_{\Xsf} + \sigma\uvec = \sigma\delta_{\Xsf} + \sigma \pi^{-1}\pi\uvec = \sigma\delta_{\xvec} + \delta_{\sigma} \pi \uvec={\color{colora}\sigma\delta_{\xvec} + \delta_{\sigma} \l\pi \uvec\r_0} + {\color{colorb}\delta_{\sigma} \l \pi \uvec\r_1}$.

If $\Xsf$'s $\type$ is $\shared$,
$\sigma\Xsf = \sigma (\l \Xsf \r_0 + \l \Xsf \r_1-\uvec + \uvec) = \sigma(\l \Xsf \r_0 + \delta_{\l\Xsf\r_1} + \uvec) = \sigma \l \Xsf \r_0 + \sigma\delta_{\l\Xsf\r_1} + \sigma \pi^{-1} \pi\uvec = \sigma \l \Xsf \r_0 + \sigma\delta_{\l\Xsf\r_1} + \delta_{\sigma}\pi \uvec={\color{colora}\sigma \l \Xsf \r_0 + \sigma\delta_{\l\Xsf\r_1} + \delta_{\sigma}\l\pi \uvec\r_0} + {\color{colorb}\delta_{\sigma}\l\pi \uvec\r_1}$.

\paragraph{Communication.}
Since $\sigma \in \Sbb_k$, $\log k$ bits are enough to represent $k$ elements.
The online phase of $\Pi_\ssp$ requires communication of $k\log k + kL$ bits (\ie, sending $\delta_\sigma, \delta_\Xsf$ in the $\raw$ case or sending $\delta_\sigma, \delta_{\l\Xsf\r_1}$ in the $\shared$ case) in $1$ round.

\paragraph{Comparison to Existing Works.}
Asharov~\etal~\cite{ccs/AsharovHIKNPTT22} spend $6kL$~bits online in three rounds.
Araki~\etal~\cite{ccs/Araki0OPRT21}'s oblivious shuffle requires $4kL$~bits in two rounds for $k$-element permutation.
The OLGA protocol~\cite{ccs/AttrapadungH0MO21} is $1$-round but communicates $2k(L + 32)$~bits offline and $2kL$~bits online.
Our $\Pi_{\ssp}$ protocol is also $1$-round, communicates $kL$~bits offline and $kL + k\log k$~bits online.
Particularly, $k$ equals to $\numedge$ or $\numnode$ for GCN.
In practice, $\log k$ is much smaller than $L$, \eg, for a $10^6$-node graph, $\log k= 20< L=64$.

\subsection{Oblivious Selection-Multiplication}
\label{sec::osm_pro}

We design the oblivious selection-multiplication 
protocol $\Pi_{\SM}$ 
in Protocol~\ref{fig:osm-main}.
It takes a private bit (called ``selector'') 
$s\in\Zbb_{2}$ from $\pp_0$ and a secret share $\l x \r$ of an arithmetic number $x\in\Zbb_{2^L}$ owned by $\pp_0$ and $\pp_1$.
$\Pi_{\SM}$ generates a secret share of $0$ if $s=0$ or share of $x$ otherwise without disclosing $s$ or $x$.

\begin{protocol}[!t]
\caption{$\Pi_{\SM}$: Oblivious Selection-Multiplication}\label{fig:osm-main}
\begin{algorithmic}[1]
	\REQUIRE $\pp_0$ inputs $(s, \l x \r_0)$ and $\pp_1$ inputs $\l x \r_1$.
	\ENSURE $\pp_0$ gets $\l sx \r_0$ and $\pp_1$ gets $\l sx \r_1$.
	\STATE {\color{gray}\COMMENT{Offline Phase: Generate Correlated Randomness}}
	\STATE $\tt, \pp_0$: Get $(b, \l u \r_0, \l bu \r_0) \leftarrow \prf(\key_0, \ctr)$
	\STATE $\tt, \pp_1$: Get $\l u \r_1 \leftarrow \prf(\key_1, \ctr)$
	\STATE $\tt$: Send $\l bu \r_1 = bu - \l bu \r_0$ to $\pp_1$
	\STATE {\color{gray}\COMMENT{Online Phase: Compute $\l sx \r$ in 1 Round}}
	\STATE $\pp_0$: Send $\delta_s= s-b$ to $\pp_1$\label{osm::masks}
	\STATE $\pp_1$: Send $\delta_{\l x\r_1} = \l x\r_1 - \l u\r_1$ to $\pp_0$\label{osm::maskx1}
	\STATE $\pp_0$: Compute $\delta_x =\l x \r_0 - \l u\r_0 + \delta_{\l x\r_1}$
	\STATE $\pp_0$: Compute $\langle sx \rangle_0 = {\color{colora}s \delta_x + \delta_s \l u\r_0 + (-1)^{\delta_s} \l bu\r_0}$\label{osm::l18}
	\STATE $\pp_1$: Compute $\langle sx \rangle_1 = {\color{colorb}\delta_s\l u\r_1 + (-1)^{\delta_s}\l bu\r_1}$\label{osm::l19}
	\RETURN $\l sx \r$
\end{algorithmic}
\end{protocol}
\begin{table}[!t]
\centering
		\caption{Communication for Oblivious Selection-Mult.}
			\setlength\tabcolsep{12pt}
			\begin{tabular}{l|c|c|c}
 		\hline
\textbf{Protocol} & \textbf{Offline} & \textbf{Online} & \textbf{Round} \\
\hline
$\promult$~\cite{crypto/Beaver91a} & $\bitlen$&$2\bitlen$& $1$\\\hline
OT~\cite{pkc/Tzeng02} &$\bitlen$&$2\bitlen + 1$&$1$\\\hline
\rowcolor{grayL}$\Pi_\SM$ & $\bitlen$&$\bitlen + 1$& $1$\\
 	 \hline 
			\end{tabular}

 $\bitlen$: bit-length of data
\label{table:osm_comm}
\end{table} 



\paragraph{Offline Phase.}
The commodity server $\tt$ assists $\pp_0, \pp_1$ to generate a random bit $b\in\Zbb_{2}$, a secret share of a random number $u\in\Zbb_{2^L}$, and correlated randomness $\l bu \r_0 \in \Zbb_{2^L}, \l bu \r_1 \in \Zbb_{2^L}$

\paragraph{Online Phase.}
$\pp_0$ masks $s$ using random $b$ to generate random $\delta_s$ (Line~\ref{osm::masks}).
$\pp_1$ masks $\l x\r_1$ using random $\l u\r_1$ to generate random $\delta_{\l x\r_1}$ (Line~\ref{osm::maskx1}).
After receiving the masked $\l x\r_1$ and $s$, $\pp_0$ and $\pp_1$ can respectively compute the shares $\l sx\r_0, \l sx\r_1$.

\paragraph{Correctness.}
Here, we 
verify that 
${\color{colora}\l sx\r_0} + {\color{colorb}\l sx\r_1}=sx$ 
by using Lemma~\ref{lem::sxb}
(proven in Appendix~\ref{sec:proof_lemma}).

\begin{restatable}{lemma}{sxb}
\label{lem::sxb}
Let $\Abb$ be an Abelian group and $\Bbb =\{0, 1\}$ be the binary group.
Let map $f: \Bbb \times \Abb \rightarrow \Abb$ be defined as $f(s, x) = x \mbox{ if } s = 1 \mbox{ else } 0$.
Then, for any $s,b\in \Bbb$ and $x,u\in\Abb$:
\begin{itemize}
	\setlength{\itemsep}{0pt}
	\setlength{\parskip}{0pt}
	\setlength{\parsep}{0pt}
\item[(i)] $f(s, x + u) = f(s, x) + f(s, u)$.
\item[(ii)] $f(s + b, x) = f(s, x) + (-1)^s f(b, x)$.
\end{itemize}
\end{restatable}

Let $f: \Bbb \times \Abb \rightarrow \Abb$ be the same $f$ as above.
Using Lemma~\ref{lem::sxb}, we have
$f(s,x) =f(s, x- u) + f(s, u) 
=f(s, \delta_x) + f(s - b + b, u ) 
=f(s, \delta_x) + f(\delta_s, u) + (-1) ^{\delta_s }f(b, u)=s\delta_x + \delta_s u + (-1)^{\delta_s} (bu)=s\delta_x + \delta_s \l u\r_0 + \delta_s \l u\r_1 + (-1)^{\delta_s} (bu)={\color{colora}s\delta_x + \delta_s \l u\r_0} + {\color{colorb}\delta_s \l u\r_1} + {\color{colora}(-1)^{\delta_s} \l bu\r_0} + {\color{colorb}(-1)^{\delta_s} \l bu\r_1}$.

\paragraph{Communication.}
$\Pi_\SM$ requires communicating $L + 1$ bits (\ie, $65$ bits for sending $ \delta_{\l x\r_1}, \delta_s$) online in $1$ round.
Except for $\Pi_\SM$, OT based~\cite{pkc/Tzeng02} protocol (by using two OT instances to select $\l x\r$ and $0$) and standard arithmetic multiplication (by transforming binary $s\in\Zbb_2$ into arithmetic $s\in\Zbb_{2^\bitlen}$) can also realize the functionality of section-multiplication.
As compared in Table~\ref{table:osm_comm}, our protocol saves about $50\%$ of communication while having the same round complexity compared to the OT-based~\cite{pkc/Tzeng02} protocol and standard Beaver-triple-based~\cite{crypto/Beaver91a} ($\promult$).

\subsection{Construction of \texorpdfstring{\osmm}{OSMM}}
\label{subsec:prosmm}
Based on our sparse matrix decomposition and protocols for OP and OSM, we present our \osmm protocol $\prosmm$ in Protocol~\ref{fig:osmm}.
It takes a sparse matrix $\adjmat \in \Mbb_{m,n}(\Rcal)$ from $\pp_0$ and a dense matrix $\feamat \in \Mbb_{n,d}(\Rcal)$ from $\pp_1$.
$\prosmm$ generates a secret share $\l \adjmat \feamat \r_i$ for $\pp_i \in \{\pp_0, \pp_1\}$ without leaking $\adjmat$ or $\feamat$.

\paragraph{\osmm Realization.}
Following Theorem~\ref{the::smm_main}, $\prosmm$ essentially performs an ordered sequence of linear transformations ($\sigma_5 \delta_m ^{\top} \gammaout \sigma_4 \Sigma ^{\top} \allowbreak\Lambda \sigma_3 \Sigma \sigma_2 \gammain \delta_n \sigma_1$) from right to left over $\pp_1$'s private input $\feamat$.
Multiplying public matrices $\delta_n, \Sigma, \Sigma^{\top}, \delta_m^T$ can be done non-interactively on secret shares 
(Lines~\ref{osmm::deltan},~\ref{osmm::Sigma},~\ref{osmm::Sigmat},~\ref{osmm::deltam}).

Permuting the rows of input $\feamat$ or intermediate output $\Ysf$ based on $\sigma_1, \ldots, \sigma_5$ are performed by invoking $d$ parallel ${\Pi}_{\ssp}$ instances as ${\Pi}_{\ssp}$ takes a column vector as input, 
but $\feamat$ and $\Ysf$ are matrices with $d$ columns (Lines~\ref{osmm::sigma1},~\ref{osmm::sigma2},~\ref{osmm::sigma3},~\ref{osmm::sigma4},~\ref{osmm::sigma5}).

Since $\gammain$ and $\gammaout$ are diagonal matrices with binary values, multiplications of them (Lines~\ref{osmm::gammain} and~\ref{osmm::gammaout}) can be done by $nd$ and $md$ parallel $\Pi_\SM$ instances, respectively.\footnote{
In practice, the parties need to pad zero values (non-interactively) before invoking the first $\Pi_\SM$ and cutting off zero values after invoking the last $\Pi_\SM$ to ensure consistent matrix dimensionality.
For simplicity, we omit this step in our $\prosmm$ protocol presentation.
}
Multiplication of $\Lambda$, a diagonal matrix with arithmetic values, is performed similarly,
but we can use the standard Beaver-triple-based multiplication protocol $\promult$ (Line~\ref{osmm::lambda}) instead of $\Pi_\SM$.


\begin{protocol}[!t]
\caption{$\prosmm$: Secure Sparse Matrix Multiplication}\label{fig:osmm}
\begin{algorithmic}[1]
	\REQUIRE $\pp_0$ inputs $\adjmat$ and $\pp_1$ inputs $\feamat$.
	\ENSURE $\pp_0$ gets $\l \Ysf \r_0$ and $\pp_1$ gets $\l \Ysf \r_1$ where $\Ysf = \adjmat\feamat$.
	\STATE $\pp_0$: Decomposes $\adjmat= \sigma_5 \delta_m ^{\top} \gammaout \sigma_4 \Sigma ^{\top} \Lambda \sigma_3 \Sigma \sigma_2 \gammain \delta_n \sigma_1$\label{osmm::decom}
	\STATE $\pp_0, \pp_1$: Invoke $\l \Ysf \r \leftarrow {\Pi}_{\ssp}(\sigma_1;\Xsf)$ \COMMENT{$\type == \raw$}\label{osmm::sigma1}
	\STATE $\pp_0, \pp_1$: Locally compute $\l \Ysf \r = \delta_n \l \Ysf \r$\label{osmm::deltan}
	\STATE $\pp_0, \pp_1$: Invoke $\l \Ysf \r \leftarrow \Pi_\SM(\gammain, \l \Ysf \r_0; \l \Ysf \r_1)$\label{osmm::gammain}
	\STATE $\pp_0, \pp_1$: Invoke $\l \Ysf \r \leftarrow {\Pi}_{\ssp}(\sigma_2, \l \Ysf \r_0; \l \Ysf \r_1)$\label{osmm::sigma2}
	\STATE $\pp_0, \pp_1$: Locally compute $\l \Ysf \r = \Sigma \l \Ysf \r$\label{osmm::Sigma}
	\STATE $\pp_0, \pp_1$: Invoke $\l \Ysf \r \leftarrow {\Pi}_{\ssp}(\sigma_3, \l \Ysf \r_0; \l \Ysf \r_1)$\label{osmm::sigma3}
	\STATE $\pp_0, \pp_1$: Invoke $\l \Ysf \r \leftarrow \promult(\Lambda, \l \Ysf \r_0; \l \Ysf \r_1)$\label{osmm::lambda}
	\STATE $\pp_0, \pp_1$: Locally compute $\l \Ysf \r = \Sigma^\top \l \Ysf \r$\label{osmm::Sigmat}
	\STATE $\pp_0, \pp_1$: Invoke $\l \Ysf \r \leftarrow {\Pi}_{\ssp}(\sigma_4, \l \Ysf \r_0; \l \Ysf \r_1)$\label{osmm::sigma4}
	\STATE $\pp_0, \pp_1$: Invoke $\l \Ysf \r \leftarrow \Pi_\SM(\gammaout, \l \Ysf \r_0; \l \Ysf \r_1)$\label{osmm::gammaout}
	\STATE $\pp_0, \pp_1$: Locally compute $\l \Ysf \r = \delta_m^\top \l \Ysf \r$\label{osmm::deltam}
	\STATE $\pp_0, \pp_1$: Invoke $\l \Ysf \r \leftarrow {\Pi}_{\ssp}(\sigma_5, \l \Ysf \r_0; \l \Ysf \r_1)$\label{osmm::sigma5}
	\RETURN $\l \Ysf \r$
\end{algorithmic}
\end{protocol}

\paragraph{Graph Protection and  Dimensions.}
All entries of $\sigma_5, \sigma_4, \sigma_3, \sigma_2, \sigma_1, \allowbreak\gammaout, \gammain,$ and $\Lambda$ are protected in $\prosmm$.
The dimensions of $\sigma_1$ and $\sigma_5$ are $\numnode$, corresponding to the number of rows in $\Asf$ and $\Xsf$. Since $\Asf$ and $\Xsf$ are held by $\pp_0$ and $\pp_1$, respectively, the dimensions of $\sigma_1$ and $\sigma_5$ are considered reasonable public knowledge.
The dimensions of $\sigma_4, \sigma_3, \sigma_2,$ and $\Lambda$ are $\numedge$, representing tolerable leakage useful for sparsity exploration.
$\gammain$ has dimensions of $\numedge \times \numnode$, while $\gammaout$ is $\numnode \times \numedge$, both of which do not incur extra graph leakage beyond $\numedge$ or $\numnode$.
Importantly, such general statistical information  about the graph does not compromise the privacy of specific nodes or edges or incur identifiable risks.

\paragraph{Correctness.}
$\prosmm$ follows the same sequence of transformations as in Theorem~\ref{the::smm_main}, which shows the correctness of our sparse matrix decomposition.
Since the underlying ${\Pi}_{\ssp}$ and $\Pi_\SM$ protocols are correct, so does
our $\prosmm$ protocol.


\paragraph{Communication.}
$\prosmm$ invokes $5$ ${\Pi}_{\ssp}$, $2$ $\Pi_\SM$, and $1$ $\promult$, communicating $2(2\numnonzero + m + n)dL$ bits offline and $((5\numnonzero + 2m + 2n)\bitlen + 3\numnonzero \log \numnonzero + m\log m + n\log n + m + n)d$ bits online in $8$ rounds (see Table~\ref{table:SMM_comm} for a breakdown).
The Beaver-triple-based protocol $\promult$~\cite{crypto/Beaver91a} for (dense) matrix multiplication communicates $mndL$ bits offline and $2mndL$ bits online.

In practical GCN usages, we have $m = n = \numnode < t = \numedge \ll \numnode^2 = mn$ and $\log m, \log n, \log t < L = 64$.
Also, $d$ is a relatively small constant.
So, the communication cost of $\prosmm$ is simplified to $O(\numedge)$, rather than $O(\numnode^2)$, by directly using $\promult$ for each entry in SMM.

\begin{table}[!t]
	\centering
	\caption{Cost for \osmm on $\adjmat \in \Mbb_{m,n}(\Rcal)$ and $\feamat \in \Mbb_{n,d}(\Rcal)$}
	\setlength\extrarowheight{2pt}
	\setlength\tabcolsep{1pt}
	\begin{tabular}{l|c|c|c}
	\hline
	\textbf{Protocol} & \textbf{Offline} & \textbf{Online} & \textbf{Rd} \\
	\hline
	\multirow{2}{*}{$5\ {\Pi}_{\ssp}$} & \multirow{2}{*}{$(3t + m + n)d\bitlen$} & $((3t + m + n)\bitlen + 3t \log t$ & \multirow{2}{*}{$5$} \\
	& & $ + m \log m + n \log n)d$ & \\
	\hline
	$2\ \Pi_\SM$ & $(m + n)d\bitlen$ & $((m + n)\bitlen + m + n)d$ & $2$ \\
	\hline
	$1\ \promult$ & $td\bitlen$ & $2td\bitlen$ & $1$ \\
	\hline
	\rowcolor{grayL} & & $((5\numnonzero + 2m + 2n)\bitlen + 3\numnonzero\log \numnonzero$ & \\
	\rowcolor{grayL}\multirow{-2}{*}{$\prosmm$} & \multirow{-2}{*}{$2(2\numnonzero + m + n)dL$} & $ + m\log m + n\log n + m + n)d$ & \multirow{-2}{*}{$8$} \\\hline
	\end{tabular}\\
		$t$: number of non-zero elements in $\adjmat$,
	$\bitlen$: bit-length of data,\\
	$d$: node feature dimensionality, Rd: round
	\label{table:SMM_comm}
\end{table}

\section{End-to-End GCN Inference and Training}
\label{sec::secgcn}

\paragraph{Implementation.}
$\cgnn$ adopts classical GCN~\cite{iclr/KipfW17} in the transductive setting with
two graph convolution layers ($\mathsf{GConv}$) 
followed by ReLU and softmax function.
We implement $\cgnn$ using Python
in the TensorFlow framework.
\begin{equation*}
\begin{aligned}
\Xsf,\Asf
	\rightarrow 
\boxed{
	\stackrel{\textstyle{\mathsf{GConv}}}{\boxed{\weimat_{1}}}
}
	\rightarrow 
 	\boxed{ {\mathsf{{ReLU}}}}
	\rightarrow 
\boxed{
	\stackrel{\textstyle{\mathsf{GConv}}}{\boxed{\weimat_{2}}}
}
	\rightarrow
	\Ysf
	\rightarrow 
 	 	\boxed{ {\mathsf{{Softmax}}}} 
 	 	\end{aligned}
 	\end{equation*}
We implemented all the above protocols (detailed in Appendix~\ref{sec::alg}) in $\cgnn$
upon secure computation with secret shares.
Since GCN inherits conventional neural networks, we still rely on similar functions and layers.
Following the ideas of \cite{sp/MohasselZ17,ccs/RatheeRKCGRS20,uss/WatsonWP22,acsac/0021ZCPTLY23,popets/AttrapadungHIKM22},
we re-implemented the relevant protocols under the 2PC setting in $\cgnn$ for ReLU,
softmax, Adam optimization, and more.
For non-sparse multiplication, $\cgnn$ still uses Beaver triples~\cite{crypto/Beaver91a}.

\paragraph{Forward Propagation.}
Recall that $\pp_0$ holds the (normalized) adjacency matrix $\Asf$ and $\pp_1$ holds the features $\Xsf$.
The first GCN layer is defined as $\Zsf =\mathsf{ReLU}(\Asf\Xsf \weimat)$, thus $\pp_0$ and $\pp_1$ jointly execute $\Pi_\smm$, $\promult$, and ReLU protocols (combining PPA~\cite{arith/Beaumont-SmithL01}, GMW~\cite{acssc/Harris03}, and $\SM$).
$\pp_0$ and $\pp_1$ will get $\l \Zsf\r_0$ and $\l \Zsf\r_1$, respectively.
Then, $\pp_0$ and $\pp_1$ securely compute $\Ysf= \Asf \Zsf \weimat$ and $\mathsf{Softmax}(\Ysf)$ in the second layer.
In the output layer, $\pp_0$ and $\pp_1$ jointly execute secure softmax protocol~\cite{acsac/0021ZCPTLY23}.

\paragraph{Backward Propagation.}
$\pp_0$ and $\pp_1$ securely compute $\mathsf{softmax} (\Ysf)-\Ysf'$ using cross-entropy loss to get $\frac{\partial loss}{\partial \Ysf}$, where $\Ysf'$ is the label matrix.
Then, we compute the gradient of a graph layer $\frac{\partial loss}{\partial \weimat}=\Zsf^{\top}\Asf^{\top} \frac{\partial loss}{\partial \Ysf}$, where the multiplication of $\Zsf^{\top}\Asf^{\top}$ is \osmm.
If we use the SGD optimizer, $\weimat$ is updated to be $\weimat \leftarrow \weimat - \eta \frac{\partial loss}{\partial W}$,
where $\eta$ is the learning rate.
If we use the Adam optimizer~\cite{iclr/KingmaB14}, $\weimat$ is updated by following the computation of Adam given $\frac{\partial loss}{\partial W}$.
The last step is to securely compute the gradient of ReLU and graph layer similarly.

\paragraph{GCN Inference and Training.} 
As for end-to-end secure GCN computations, $\pp_0$ and $\pp_1$ collaboratively execute a sequence of protocols to run a single forward propagation (for inference) or forward and backward propagation iteratively (for training).
$\cgnn$ supports both single-server simulation for multiple hosts and multiple-server execution in a distributed setting.
Using $\cgnn$, researchers and practitioners can realize various GCNs
using the template of $\texttt{class SGCN}$ (Appendix~\ref{sec::alg}),
similar to using the TensorFlow framework except that \emph{all computations are over secret shares}.

\section{Experiments and Evaluative Results}

We evaluate the performance of our \osmm protocol and $\cgnn$'s private GCN inference/training on three Ubuntu servers with $16$-core Intel(R) Xeon(R) Platinum 8163 2.50GHz CPUs of $62$GB
RAM and NVIDIA-T4 GPU of $16$GB RAM.
We aim to answer the three questions below.

\noindent\textbf{Q1.} \emph{How much communication/memory-efficient and accurate for \cgnn?} (\S\ref{sec::comm_compare_gcn}, \S\ref{sec::smmmemory}, \S\ref{sec::acc_compare_gcn})

\noindent\textbf{Q2.} \emph{How do different network conditions impact the running time of \cgnn's inference and training?} (\S\ref{sec::time_net})

\noindent\textbf{Q3.} \emph{How much efficiency has been improved by \osmm?} (\S\ref{sec:ablation})


{Graph Datasets.}
We consider three publication datasets widely adopted in GCN training: Citeseer~\cite{dl/GilesBL98}, Cora~\cite{aim/SenNBGGE08}, and Pubmed~\cite{ijcnlp/DernoncourtL17}.
Their statistics are summarized in Table~\ref{tab:datasets}.
\begin{table}[!t]
\centering
\caption{Dataset Statistics}
\setlength\tabcolsep{2pt}
\begin{tabular}{l|rrrr|rr}
\hline
\multicolumn{1}{c|}{\textbf{Dataset}} & \textbf{Node} & \textbf{Edge} & \textbf{Feature} & \textbf{Class} &\textbf{\# Train} &\textbf{\# Test}
\\ \hline
Cora 	 & $2,708$ 	 & $5,429$ 	 & $1,433$ & $7$ & $140$ & $1,000$ \\
Citeseer 	 	 & $3,312$ 	 & $4,732$ 	 & $3,703$ 	 & $6$ & $120$ & $1,000$ \\
Pubmed 	 	 & $19,717$ 	 	 & $44,338$ 	 	 & $500$ 	 & $3$ 	 & $60$ & $1,000$ \\
\hline
\end{tabular}
\begin{tablenotes}
\item \# Train/Test: number of samples in training/test dataset
\end{tablenotes}
\label{tab:datasets}
\end{table}

\subsection{Communication of GCN}
\label{sec::comm_compare_gcn}
To evaluate communication costs in $\cgnn$, we record the transmitting data, including frame and MPC-related data in both online and offline phases, across the servers or ports.
The inference refers to
a forward propagation, while the training involves an epoch of training.
Unlike classical CNN training over independent data points, GCN training feeds up the whole graph (\ie, $1$ batch) in each training epoch, thus no benchmarking for batch sizes.

\paragraph{Secure Training.} 
SecGNN~\cite{tsc/WangZJ23} and  CoGNN~\cite{ccs/ZouLSLXX24} are the only two open-sourced works for secure training with MPC.
SecGNN~\cite{tsc/WangZJ23} is the first work, meanwhile
 CoGNN~\cite{ccs/ZouLSLXX24} and its optimized version CoGNN-Opt are the most recent advances.
Thus, we choose them as \cgnn's counterparts for comparison.
Table~\ref{table:comm_on_off_gcn} shows their comparison results.
In general, $\cgnn$ uses ${\leq}1.3$GB in all cases.
Using SGD, $\cgnn$ uses $0.3075$GB, $0.5400$GB, and $1.2567$GB for training over Cora, Citeseer, and Pubmed.
With Adam, $\cgnn$ costs slightly higher communication due to SGD not needing $1/\sqrt{x}.$
The additional
costs are $6.2\%$, $3.7\%$, and $0.8\%$ for
training.
These differences are related to the sparsity of
data and the times of gradient update.
All the cases above require less communication costs than CoGNN and SecGNN.

\paragraph{Secure Inference.}
Except for CoGNN and SecGNN, we additionally compare \cgnn with the most recent secure inference  work -- OblivGNN~\cite{uss/XuL0AYY24} for comprehensiveness.
Table~\ref{table:comm_on_off_inf} compares the communication costs.
$\cgnn$ requires the lowest communication costs in all cases, reducing by  $\sim50\%$ of OblivGNN and $\sim 80\%$ of CoGNN-Opt.

\begin{table}[!t]
 	\centering
 	\caption{Communication (GB/epoch) for Training}
 	\label{table:comm_on_off_gcn}
 \setlength\tabcolsep{9pt}
 	\begin{tabular}{l|ccc}
 	\hline
 	\multirow{2}{*}{\textbf{Framework}} & \multicolumn{3}{c}{\textbf{Dataset}}
	\\\cline{2-4}
 & \textbf{Cora} 	 & \textbf{Citeseer} & \textbf{Pubmed}\\\hline
 	SecGNN & 	$18.99$ & $48.21$ & $31.74$\\
    CoGNN & $86.99$ & $202.81$ &$273.25$ \\
    CoGNN-Opt & $0.82$ & $1.4$ & $4.33$\\\hline
 \rowcolor{grayL}$\cgnn$ (SGD) & $0.3075$ & $0.5400$ & $1.2567$ \\
 \rowcolor{grayL}$\cgnn$ (Adam) & $0.3265$ & $0.5600$ & $1.2667$ \\
 	 	\hline
 	 	\end{tabular}
 	\end{table}

\begin{table}[!t]
 	\centering
 	\caption{Communication for Inference}
 	\label{table:comm_on_off_inf}
\setlength\tabcolsep{11pt}
 	\begin{tabular}{l|ccc}
 	\hline
 	\multirow{2}{*}{\textbf{Framework}} & \multicolumn{3}{c}{\textbf{Dataset}}
	\\\cline{2-4}
 & \textbf{Cora} 	 & \textbf{Citeseer} & \textbf{Pubmed}\\\hline
	SecGNN & 	$1$GB & $1.7$GB & $2.5$GB\\
    CoGNN & $85.63$GB & $201.29$GB &$263.59$GB\\
 CoGNN-Opt & $0.5$GB &$0.91$GB & $2.02$GB\\\hline
 OblivGNN-B & $34.32$GB & $61.81$GB &$16.33$GB\\
 OblivGNN &$0.29$GB&$0.41$GB& $1.65$GB\\\hline
    \rowcolor{grayL}$\cgnn$ & $114$MB & $274$MB & $602$MB \\
 	 	\hline
 	 	\end{tabular}
\end{table}

\subsection{Memory Usage}
\label{sec::smmmemory}
To avoid extra irrelevancy (\eg, communication), we tested the memory usage
on a single server, recording the largest observed value.
Table~\ref{table:mem_smm} reports memory usage for training with $\prosmm$ and the standard $\promult$ using Beaver triple.
Both protocols show acceptable results for
smaller Cora and Citeseer datasets.
Yet, $\prosmm$ saves $14.5\%, 20.8\%$ memory with secure SGD,
and $10.5\%, 18.2\%$ memory with secure Adam.

The maximum memory SGD training uses is slightly lower than with Adam, as Adam's optimization requires more memory.
When training over the larger Pubmed dataset,
an out-of-memory (OOM) error occurs (marked by $^{\oom}$) when using Beaver triple, whereas the \cgnn with $\prosmm$ supports the stable use (${<}2.7$GB) of memory for all datasets.

\begin{table}[!t]
 	\caption{Maximum Memory Usage (GB) for Training}
 	\label{table:mem_smm}
 	\setlength\tabcolsep{3pt}
 	\centering
 	\begin{tabular}{c|c|c|c|c}
 	\hline
	\textbf{Optimizer} & 	\textbf{Dataset} & \textbf{Protocol} & \textbf{Memory} & \textbf{Reduction} 
	\\
 	\hline 
	\multirow{6}{*}{\textbf{SGD}} &\multirow{2}{*}{\textbf{Cora}} & Beaver & $1.31$ &
	\\
 	 & & \cellcolor{grayL}$\prosmm$ & 	\cellcolor{grayL}$1.12$ & \cellcolor{grayL}$14.5\%$
 	\\\cline{2-5}
 	 & \multirow{2}{*}{\textbf{Citeseer}} & Beaver & $2.07$ &
	\\
 	 & & \cellcolor{grayL}$\prosmm$ & \cellcolor{grayL}$1.64$ & \cellcolor{grayL}$20.8\%$ 
	\\\cline{2-5}
	 & \multirow{2}{*}{\textbf{Pubmed}} & Beaver & ${>}28.82^{\oom}$ & 
	\\
	 & & \cellcolor{grayL}$\prosmm$ & \cellcolor{grayL}$1.94$ & \cellcolor{grayL}${>}93.3\%$
	\\\hline
 	\multirow{6}{*}{\textbf{Adam}} & \multirow{2}{*}{\textbf{Cora}} & Beaver & $1.91$ &
	\\
 	 & & \cellcolor{grayL}$\prosmm$ & 	\cellcolor{grayL}$1.71$ & \cellcolor{grayL}$10.5\%$
	\\\cline{2-5}
 	 & \multirow{2}{*}{\textbf{Citeseer}} & Beaver & $2.75$ &
	\\
 	 & & \cellcolor{grayL}$\prosmm$ & 	\cellcolor{grayL}$2.25$ & \cellcolor{grayL}$18.2\%$ 
	\\\cline{2-5}
	 & \multirow{2}{*}{\textbf{Pubmed}} & Beaver & ${>}28.02^{\oom}$ & 
	\\
	 & & \cellcolor{grayL}$\prosmm$ & \cellcolor{grayL}$2.69$ & \cellcolor{grayL}${>}90.4\%$ 
	\\\hline
 \end{tabular}
 \begin{tablenotes}
 \item  ${\oom}$:  out-of-memory (OOM) error occurs.
 \end{tablenotes}
\end{table}

\subsection{Model Accuracy}
\label{sec::acc_compare_gcn}
We trained the GCN over different datasets from random initialization for $300$ epochs using Adam~\cite{iclr/KingmaB14} with a $0.001$ learning rate.
Our configuration of model parameters (\ie, the dimensionality of hidden layers and the number of samples) follow the original setting~\cite{iclr/KipfW17}.
Since model accuracy is meaningful only for identical partitioning strategy, we compare the accuracy of secure training with plaintext in the same contexts in Table~\ref{tab:acc_inf_tra}.
Our results show that $\cgnn$'s accuracy is comparable to that of plaintext~training.
Specifically, $\cgnn$ achieves $\{73.5\%, 64.4\%, 75.4\%\}$ after $100$ epochs and $\{76.0\%, 65.1\%, 75.2\%\}$ after $300$ epochs.
Due to fluctuated training convergence, fixed-point representation, and non-linear approximation, model accuracy is slightly different.

 
\begin{table}[!t]
 	\centering
 	\caption{Model Accuracy}
 	\label{tab:acc_inf_tra}
 	\setlength\tabcolsep{7pt}
 	\begin{tabular}{l|ccc}
	\hline
 	\multirow{2}{*}{\textbf{Framework}} & \multicolumn{3}{c}{\textbf{Dataset}}
	\\\cline{2-4}
 & \textbf{Cora} &\textbf{Citeseer} & \textbf{Pubmed}\! \\
 	 	\hline
	{Plaintext}
	& $75.7\%$ 	 & $65.4\%$ & $74.5\%$\! \\
 	\rowcolor{grayL}$\cgnn$ ($100$ Epochs)\! & $73.5\%$ & $64.4\%$ & \textbf{75.4\%}\! \\
 	\rowcolor{grayL}$\cgnn$ ($300$ Epochs)\! & \textbf{76.0\%} & \textbf{65.1\%} & $75.2\%$\! \\
 	 	\hline
 	 	\end{tabular}
\end{table}

\subsection{Running Time in Different Networks}
\label{sec::time_net}
We simulate real-world deployment under different network conditions for \osmm, private inference, and private training.
In particular, we consider a normal network condition ($800$Mbps, $0.022$ms) and two poor network conditions, including a narrow-bandwidth (N.B.) network ($200$Mbps, $0.022$ms) and high-latency (H.L.) network ($800$Mbps, $50$ms).
Additionally,  TCP transmission involves the process of three-step handshake, data transmission, congestion control, and connection termination, thus practical time delay of $\osmm$ is varied below under different network conditions.
\begin{table*}[!t]
	\centering
		\caption{$1$-Epoch Training Time (seconds) in Normal, Narrow-Bandwidth, or High-Latency Networks}
		\label{table:tra_net}
  \setlength\tabcolsep{8pt}
			\begin{tabular}{l|l|r|r|r|r|r|r}
				\hline
     \multirow{2}{*}{\textbf{Dataset}} & \multirow{2}{*}{\textbf{{Protocol}}} & \multicolumn{2}{c|}{\textbf{Normal ($800$Mbps, $0.022$ms)}}& \multicolumn{2}{c|}{\textbf{N.B. ($200$Mbps, $0.022$ms)}}& \multicolumn{2}{c}{\textbf{H.L. ($800$Mbps, $50$ms)}}
     \\\cline{3-8}
     & & \phantom{sherman}{\textbf{SGD}} 
     & {\textbf{Adam}}
     & \phantom{sherman}{\textbf{SGD}}
     & \phantom{123}{\textbf{Adam}}
     & \phantom{sherman}{\textbf{SGD}}
     & \phantom{123}{\textbf{Adam}}\\
				\hline
    \multirow{3}{*}{\textbf{Cora}} & Beaver &$6.55 $ &$7.89 $  &$25.98 $   &$27.55 $   &$11.70 $  &$19.57 $  \\
	  &$\prosmm$&$4.20 $  &  $ 5.55$  &  $13.29 $  &  $14.88 $  
   & \raggedleft $9.11$
   &  $16.72 $   \\ 
     & (Saving) 
     & \cellcolor{grayL}{$35.9\%$} 
     & \cellcolor{grayL}{$29.7\%$} 
     & \cellcolor{grayL}{\textbf{48.8\%}} 
     & \cellcolor{grayL}{\textbf{46.0\%}} 
     & \cellcolor{grayL}{$22.1\%$}
     & \cellcolor{grayL}{$14.6\%$}
     \\\hline

   \multirow{3}{*}{\textbf{Citeseer}} & Beaver &$11.66 $   &$ 13.20$  &$46.31 $  &$48.75 $  &$ 18.53$  &$27.93 $  \\
   
    & $\prosmm$
    & \raggedleft $6.77$ 
    & \raggedleft $8.35$  & $24.00 $  & $26.44 $  & $ 13.47$  & $ 21.26$  \\
      & (Saving) & \cellcolor{grayL}{$41.9\%$} &\cellcolor{grayL}{$36.7\%$} & \cellcolor{grayL}{\textbf{48.2\%}} & \cellcolor{grayL}{\textbf{45.8\%}}& \cellcolor{grayL}{$27.3\%$}&\cellcolor{grayL}{$23.9\%$} \\\hline
    
     \multirow{2}{*}{\textbf{Pubmed}} & Beaver &OOM  &OOM& OOM& OOM& OOM& OOM  \\    
 & $\prosmm$  &$ 22.87 $   &$24.45 $ & $63.69 $   &$63.58 $   &$32.00 $  &$39.86 $ \\
	\hline
	\end{tabular} 
\end{table*}

\paragraph{Inference Time.}
Table~\ref{table:inf_time_net} compares the private inference time,
including TensorFlow-Graph construction and forward-propagation computation of GCN in varying network conditions
over multiple
datasets.
Compared to adopting Beaver triples, $\cgnn$ via $\prosmm$ is ${\sim}7\%$-$19\%$ faster in the normal network, ${\sim}35\%$-$45\%$ quicker in the narrow-bandwidth one, and saves ${\sim}6\%$-$17\%$ time in the high-latency setting.
The OOM problem prevents us from evaluating inference over Pubmed using Beaver triples, 
while $\prosmm$ takes ${\sim}30$-$50$s.

\begin{table}[!t]
    \centering
        \caption{Inference Time (seconds) in Varying Networks}
        \label{table:inf_time_net}
  \setlength\tabcolsep{5pt}
            \begin{tabular}{c|c|r|r|r}
                \hline
     \textbf{Dataset} & \textbf{Protocol} & \textbf{Normal} & \textbf{N.B.} & \textbf{H.L.}\\
                \hline
    \multirow{3}{*}{\textbf{Cora}} & Beaver &$17.48$  &  $28.34$ &   $24.22$    \\
       &$\prosmm$&$16.27 $  &  $ 21.06$  &  $22.76 $   \\ 
       & (Saving) &{$6.9\%$}& \cellcolor{grayL}{$25.7\%$} & {$ 6.0\%$}  \\\hline
   \multirow{3}{*}{\textbf{Citeseer}}& Beaver & $24.57 $    &$ 44.39$  &$33.32$    \\
     &$\prosmm$&$ 20.58 $    &$30.57 $    &$28.49 $     \\ 
     & (Saving) & \cellcolor{grayL}{$16.2\%$} & \cellcolor{grayL}{$31.3\%$} & \cellcolor{grayL}{$14.5\%$}   \\\hline
     \multirow{2}{*}{\textbf{Pubmed}} & Beaver & OOM  &  OOM    & OOM  \\
 & $\prosmm$ &$ 29.38$     &$49.93 $  &$38.40 $  \\
    \hline
            \end{tabular}

N.B.: narrow bandwidth,
H.L.: high latency
\end{table}


\paragraph{Training Time.}
Table~\ref{table:tra_net} compares the private training time
with SGD/Adam in varying network conditions over different datasets.
We tested $10$ epochs and got the average.
In the normal network, $\cgnn$ via $\prosmm$ is ${\sim}56\%$-$73\%$ faster with SGD and ${\sim}42\%$-$58\%$ faster with Adam.
In the narrow-bandwidth network, $\cgnn$ via $\prosmm$ is ${\sim}93\%$-$95\%$ quicker with SGD and ${\sim}84\%$-$85\%$ quicker with Adam.
Besides, $\cgnn$ via $\prosmm$ is ${\sim}28\%$-$38\%$ faster with SGD and ${\sim}17\%$-$32\%$ faster with Adam in the high-latency setting.

\begin{table}[!t]
 	\centering
 	\caption{Communication Costs (MB) for SMM}
 	\label{table:comm_smm}
 	\setlength\tabcolsep{7pt}
 	 	\begin{tabular}{c|c|r|r|r}
 	 	\hline
 	\textbf{$\#$E/N} & \textbf{$\#$Node} & \textbf{Beaver} & ${\prosmm}$ & \textbf{Saving} \\
 	 	\hline
 	$1$ & $1000$ & $25.1$ & $0.8$ & \cellcolor{grayL}{$ 96.8\%$} \\
 	$1$ & $2000$ & $ 100.3$ & $1.3$ & \cellcolor{grayL}{$ 98.7\%$}\\
 	$1$ & $5000$ & $626.1$ & $2.8$ & \cellcolor{grayL}{\textbf{99.6\%}} 	\\
 	\hline
 	$2$ & $1000$ & $25.1$ & $1.0$ & \cellcolor{grayL}{$ 95.9\%$} \\
 	$2$ & $2000$ & $100.3$ & $1.8$ & \cellcolor{grayL}{$ 98.2\%$} \\
 	$2$ & $5000$ & $626.0$ & $3.9$ & 	\cellcolor{grayL}{\textbf{99.4\%}} \\
 	\hline
 	$3$ & $1000$ & $25.1$ & $1.3$ & \cellcolor{grayL}{$ 95.0\%$} \\
 	$3$ & $2000$ & $100.3$ & $2.2$ & \cellcolor{grayL}{$ 97.8\%$} \\
 	$3$ & $5000$ & $626.1$ & $5.1$ & \cellcolor{grayL}{\textbf{99.2\%}} 	\\
 	\hline
 	 	\end{tabular} 
 	 	
 	$\#$E/N: ratio of edges per node,
	\\``Beaver'': using Beaver triples for SMM
\end{table}

\begin{table}[!t]
 	\centering
 	\caption{Memory Usage (MB) Given Varying Sparsity}
 	\label{table:mem_smm_sparse}
 \setlength\tabcolsep{7pt}
 	 	\begin{tabular}{c|c|r|r|r}
 	 	\hline \textbf{$\#$E/N} & \textbf{$\#$Node} & \textbf{Beaver} & ${\prosmm}$ & \textbf{Saving} \\
 	 	\hline
 	$1$ & $1000$ & $688.6$ & $572.7$ & \cellcolor{grayL}{$ 16.83\%$} \\
 	$1$ & $2000$ & $1236.5$ & $575.9 $ & \cellcolor{grayL}{$ 53.42\%$}\\
 	$1$ & $5000$ & $2136.0$ & $ 583.7$ & \cellcolor{grayL}{\textbf{72.67\%}} \\
 	\hline
 	$2$ & $1000$ & $680.6$ & $575.1 $ & \cellcolor{grayL}{$ 15.50\%$} 	\\
 	$2$ & $2000$ & $1173.4$ & $ 579.8$ & \cellcolor{grayL}{$ 50.59\%$} \\
 	$2$ & $5000$ & $2135.8$ & $596.1 $ & 	\cellcolor{grayL}{\textbf{72.09\%}} \\
 	\hline
 	$3$ & $1000$ & $719.2$ & $ 578.5$ & \cellcolor{grayL}{$ 19.56\%$} 	\\
 	$3$ & $2000$ & $1142.0$ & $582.2 $ & \cellcolor{grayL}{$ 49.02\%$} \\
 	$3$ & $5000$ & $2136.8$ & $ 605.3$ & \cellcolor{grayL}{\textbf{71.67\%}} \\
 	\hline
 	 	\end{tabular}
\end{table}
\subsection{Ablation Study for \texorpdfstring{\osmm}{OSMM}}
We perform extensive experiments 
\label{sec:ablation}
to study the computational, communication, and memory costs saved by \osmm.
Due to saving space, we defer some experimental results to Appendix~\ref{sec::more_exp}.

\paragraph{Communication.} 
\label{sec::smmcomm}
Table~\ref{table:comm_smm} reports communication comparison given varying sparsity of matrices with ${\sim}1000$-$5000$ nodes, each with $\{1, 2, 3\}$ edges on average.
In a training epoch, Beaver triples cost ${\sim}25$-$626$MB for sparse MM, whereas $\prosmm$ spends relatively stable costs of roughly
\mbox{${\sim}1$-$5$MB.}
$\prosmm$ reduces $95\%+$ communication compared with standard MM in all cases.
At best, $\prosmm$ costs only $0.4\%$ communication of standard one when $\#$Node is $5000$ with $1$ edge on average (also the sparsest case in Table~\ref{table:comm_smm}).


\paragraph{Memory Usage.}
Table~\ref{table:mem_smm_sparse} shows how sparsity affects memory usage.
$\#$Node is the total number of nodes and $\#$Edge/Node is the average number of edges connected per node.
Memory usage via Beaver triples scales with $\#$Node, whereas $\prosmm$ maintains relatively stable use.
In detail, $\prosmm$ reduces ${\sim}16\%$-$73\%$ memory for ${\sim}1000$-$5000$ nodes.

\paragraph{Running Time under Varying Network Conditions.}
Table~\ref{table:time_smm_net_10} reports the running time of $10$ epochs of \osmm in the normal, narrow-bandwidth, and high-latency networks.
In the normal network, $\prosmm$ achieves a ${\sim}1.1$-$26.3\times$ speed-up compared with Beaver triples.
In the narrow-bandwidth network, $\prosmm$ is ${\sim}1.53$-$41.96\times$ faster, showing a higher speed-up than the normal network.
In the high-latency network, $\prosmm$ shows a slightly lower speed-up than the normal network.
 The reason is that $\prosmm$ uses more rounds of communication than Beaver triples.
 It would be interesting to explore reducing round complexity for \osmm in the future.


\paragraph{Running Time with Varying Dimensionality.}
In practice, feature dimensionality (\eg, salary, life cost) is not very high.
We vary it across $\{10, 20, 50\}$ over the Citeseer dataset in Figures~\ref{fig:time_fea_dim_inference_cite},~\ref{fig:time_fea_dim_train_cite} (results for other datasets are in Appendix~\ref{sec::more_exp}).
We test both inference and training times.
The fewer feature dimensions, the higher the percentage of costs is from SMM.
Roughly, the time costs have been reduced by ${\sim}50$-$75\%$.

\begin{figure}[!t]
	\centering
	\includegraphics[width = 0.42\textwidth]{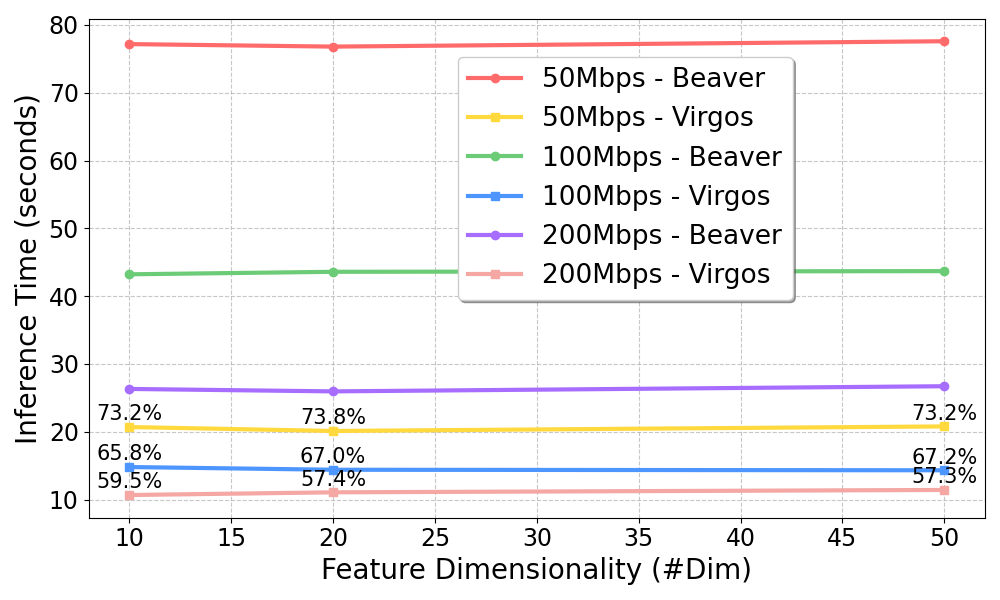}
		\caption{Inference Time with Feature Dimensionality}
 	\label{fig:time_fea_dim_inference_cite}
\end{figure}
\begin{figure}[!t]
	\centering
	\includegraphics[width = 0.42\textwidth]{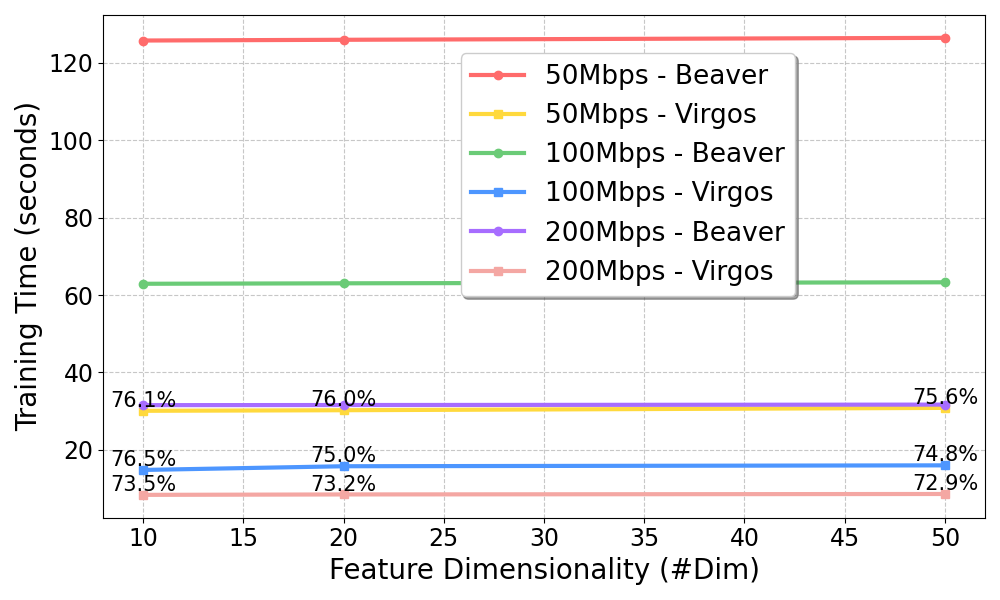}
		\caption{Training Time with  Feature Dimensionality}
 	\label{fig:time_fea_dim_train_cite}
\end{figure}


\section{Related Works}
\label{supp::work}
 
\paragraph{Privacy-Preserving Machine Learning (PPML) with MPC}
In the past decade, PPML has gained great attention as establishing a well-performing neural network often requires massive sensitive data, \eg, human faces, medical records.
Cryptography, especially MPC,
provides a handy tool to 
hide all inputs and intermediate results from adversaries.
Secure computation of various operations~\cite{ndss/ABY15,ccs/ABY318,neurips/crypten2020,ccs/RatheeRKCGRS20,sp21/TanKTW,uss/WatsonWP22,acsac/0021ZCPTLY23,adma/ZhengSDCLZW23} like softmax
and ReLU has been realized efficiently.
GPU-friendly frameworks/libraries like 
CryptGPU~\cite{sp21/TanKTW} and Piranha~\cite{uss/WatsonWP22} have been also proposed.
They have shown good computational performance in training CNNs.
However, most works are not tailored for 
GCNs, especially those computation over large and sparse structures.

\paragraph{Secure (Dense) Matrix Multiplication}
Classical secret-sharing schemes produce secret-shares to dense matrix multiplication.
Recent works~\cite{sp/MohasselZ17,asiaccs/MonoG23,asiacrypt/0030KRRSW20,ccs/JiangKLS18} designed more efficient protocols to reduce communication costs.
Yet, directly adapting these to sparse structures still results in high memory/communication costs asymptotically growing with the matrix size.
Large communication overhead persists as a major concern in PPML, \eg, consuming $94\%$ of the training time of Piranha~\cite{uss/WatsonWP22}.
Even worse, large matrix computations are not supported due to memory overflow.
So, minimizing communication costs and memory usage of \osmm is crucial.

Standard matrix decomposition methods, such as singular value decomposition (SVD)~\cite{banerjee2014linear} and LU decomposition (LUD)~\cite{bunch1974triangular}, are designed for faster plaintext computations rather than reducing communication in secure MPC.
Thus, employing these decompositions in secure GCN does not significantly lower communication costs.
Specifically, SVD decomposes a matrix into dense and diagonal matrices, while LUD decomposes it into triangular matrices. 
They both require $O(\numnode^2)$ communication for secure multiplications in GCNs.
Our decomposition approach instead adapts the graph topology into a sequence of linear transformations to exploit the sparsity, finally achieving $O(\numedge)$ communication.

\paragraph{Sparsity Exploration in MPC} 
Exploiting the sparsity in plaintext can speed up the computation.
Directly encoding the input sparse matrices into random matrices for acquiring privacy destroys the 
sparsity~\cite{isit/XhemrishiBW22}.
Several recent works~\cite{isit/XhemrishiBW22,tifs/BitarEWX24} studied the secret-shared sparse matrices and their multiplication by bridging the trade-offs between sparsity and privacy.
Specifically, they relax the privacy constraint by focusing on multiplying a secret-shared matrix with a public matrix.
Unlike their works, we work in the classical MPC settings, where all inputs are secretly shared. We also represent the sparsity through algebra relations (without destroying the sparsity).

ROOM~\cite{ccs/SchoppmannG0P19} presents three instantiations of sparse matrix-vector multiplications optimized for different sparsity settings, such as row- and column-sparse
matrices.
Our decomposition works on arbitrary-sparse matrices, in contrast to either row- or column-sparse as in ROOM~\cite{ccs/SchoppmannG0P19}, hence eliminating the need to know the sparsity types for the input matrices.
Chen~\etal~\cite{kdd/0001ZWWFTWLWH21} realize \osmm 
by sending a homomorphically encrypted dense matrix to the party holding the sparse matrix to perform ciphertext multiplication and split the result into secret shares.
The limited support of homomorphic multiplication curbs 
this approach.

\paragraph{Relevant Primitives for \osmm}
Using oblivious shuffle to realize \osmm demands $O(\numnode)$ rounds of $\ssp$.
We aim to hide the corresponding permutations directly in $O(1)$ rounds 
without using oblivious shuffle~\cite{ndss/SongYBDC23,uss/JiaSZDG22} or sorting~\cite{ccs/AsharovHIKNPTT22}.
OLGA~\cite{ccs/AttrapadungH0MO21} achieves $\ssp$ as a special case of linear group action.
Its subsequent work~\cite{ccs/AsharovHIKNPTT22} also uses 
replicated secret sharing.
In $\cgnn$, private permutation operation is owned by the graph holder, leading to better efficiency as in Table~\ref{table:op_comm}.
Zou~\etal \cite{ccs/ZouLSLXX24} design a permutation protocol by packing the permutation-relied computation into offline to optimize the online communication.
Differently, we make the offline phase independent of derived permutation, thus promisingly enabling varying graph setting (a.k.a inductive training~\cite{nips/WuYY21}).

Another primitive of Oblivious Selection-Multiplication ($\SM$) is to obliviously indicate whether message passing exists in an edge.
Previous works like Multiplexer~\cite{ccs/RatheeRKCGRS20} and binary-arithmetic multiplication~\cite{ndss/ABY15} can be adopted to realize the $\SM$'s functionality.
Multiplexer~\cite{ccs/RatheeRKCGRS20}, realized by two instances of $1$-out-of-$2$ OT, requiring $2(2\bitlen + 1)$ bits.
Binary-arithmetic conversion communicates $\bitlen(\bitlen+1)/2$ bits for a $64$-bit data. 
However, $\cgnn$'s $\SM$ protocol requires $\bitlen+1$ bits online in $1$ round using secret sharing (free of logarithm rounds of combining OT).


\begin{table}[!t]
\centering
\caption{MPC Frameworks for Secure Graph Learning}
\setlength\tabcolsep{2pt}
\begin{tabular}{|l|r|c|c|c|}
\hline
\multicolumn{1}{|c|}{\textbf{Framework}} & \textbf{Scenario} & \textbf{Inference} & \textbf{Training} & \textbf{Security} 
\\ \hline\hline
 OblivGNN~\cite{uss/XuL0AYY24}	 	 & MLaaS	 & $\checkmark$ 	 & $\times$	 & Semi-honest  \\\hline
LinGCN~\cite{nips/PengRLZHTGWXWD23} 	 & MLaaS	 & $\checkmark$ 	 & $\times$	 & Semi-honest  \\\hline
Penguin~\cite{nips/RanXLWQW23} 	 & MLaaS	 & $\checkmark$ 	 & $\times$	 & Semi-honest  \\\hline
 CoGNN~\cite{ccs/ZouLSLXX24} 	 & Horizontal 	 & $\checkmark$ 	 & $\checkmark$ & Semi-honest\\\hline
\cgnn 	 	 & Vertical	 	 & $\checkmark$ 	 	 & $\checkmark$ 	 & Semi-honest 	 \\
\hline
\end{tabular}
\label{tab:compare_sec_graph}
\end{table}

\paragraph{Secure Graph Analysis.} 
Secure graph analysis~\cite{sp/NayakWIWTS15,ccs/Araki0OPRT21} can be adopted for \osmm by reversing the graph analysis process (depicted by arrows in Figure~\ref{graph_psigqx_diff}).
Building on  garbled circuits, GraphSC~\cite{sp/NayakWIWTS15} uses an oblivious sorting to enable secure graph computation for message-passing algorithms.
Garbled circuits, while providing constant round complexity, are known for their communication and computation-intensive costs.
To address this, Araki~\etal~\cite{ccs/Araki0OPRT21} improve it by replacing sorting with  shuffling in the message-passing phase 
and use secret-sharing-based techniques to reduce the costs of communication and computation.
Recently,  Graphiti~\cite{ccs/KotiKPG24}, an advancement over GraphSC, optimizes the round complexity independently of graph size, enhancing scalability for large graphs.

In contrast to aforementioned MPC-based partitioning settings, $\cgnn$ adopts a practical vertical partitioning approach.
$\cgnn$ is designed to optimize MPC protocols specifically for vertically partitioned data, incorporating novel sparsity decomposition techniques, as well as new permutation and selection multiplication protocols. 
These innovations allow \cgnn to yield the optimal round complexity  for \osmm and be independent of the number of graph size, thus highlighting the potential  to scale and train on massive graphs under real-world network conditions.

\paragraph{Cryptographic Learning over Graphs.}
Table~\ref{tab:compare_sec_graph} summarizes recent advances for cryptographic graph learning. 
SecGNN~\cite{tsc/WangZJ23} is the first try to realize secure GCN training 
by integrating existing PPML advances.
Efficient GCN training still remains barren nowadays.
Without the customized MPC protocol for \osmm, 
SecGNN suffers from high communication costs in practice.
As for federated training, both vertical and horizontal partitions of distributed data are vital for practical usage.
Very recently, CoGNN~\cite{ccs/ZouLSLXX24} considers a collaborative training setting where each pair of computing parties knows the sub-graphs for secure training.
$\cgnn$ considers that one party who knows the graph but not the associated data, which is vertically partitioned.

Another branch of works~\cite{nips/PengRLZHTGWXWD23,nips/RanXLWQW23,nips/RanWGYXW22,uss/XuL0AYY24} adopt machine learning as a service (MLaaS) to realize secure GCN inference.
Penguin~\cite[Table 3]{nips/RanXLWQW23}, as the-state-of-art work, largely reduces the inference latency~\cite{uss/JuvekarVC18} by $5.9\times$ over the Cora dataset~\cite{aim/SenNBGGE08}, finally reaching $10$ minutes for the inference.
OblivGNN~\cite{uss/XuL0AYY24} recently reduces it to about $2$ minutes.
Unlike secure inference for MLaaS, $\cgnn$ made a noticeable step of secure training by reducing communication costs to $114$MB in roughly $20$s over the same dataset.

Many works focus on different privacy (DP) guarantees
~\cite{tnn/WuPCLZY21,ccs/SajadmaneshG21,pvldb/PatwaSGMR23},
or 
applying HE and private information retrieval for secure social recommendation~\cite{nips/CuiCLYW21}.
Like some prior works, sparsity is also exploited.
Their technical contributions differ vastly from ours, for we consider 
the MPC settings.

\section{Future Works}
\paragraph{Practical Scenarios and Graph Partitions}
\label{sec:future}
Our work can be extended to other federated learning settings, where a set of parties hold different types of features and sub-graphs in a general case.
We observe that the general case can be formulated as partitioning node features \( \Xsf_i \) and sub-matrices \( \Asf_{ij} \) across multiple parties \( \pp_i,\pp_j \).
To optimize the practical efficiency, each pair of  parties can parallel execute $\cgnn$ to compute $\Asf_{ij}\Xsf_{i}$ securely.
In future practical applications, researchers can streamline the hybrid MPC protocols by integrating plaintext handcrafts, 
leveraging secure computation as a pragmatic alternative for cross-organizational collaboration.

\paragraph{Modular Design and Security Models}
 New protocols in \cgnn are designed with a modular sense, allowing it to be instantiated with different MPC protocols. 
 Accordingly, the security model offered by the choice of MPC protocols will be carried forward, enabling \cgnn adapting to different settings.
Besides, $\ssp$ and $\SM$ are modular protocols, which may serve as building blocks of future MPC construction.
New block designs can be extended to other realizations, \eg, 2PC protocols from homomorphic encryption~\cite{eurocrypt/LyubashevskyPR10} or oblivious transfer~\cite{pkc/Tzeng02,ccs/KellerOS16} or oblivious shuffle~\cite{asiacrypt/ChaseGP20,ndss/SongYBDC23}.

\paragraph{Different GNN Models}
Besides GCN applications, our general theorem of sparse matrix decomposition 
holds potential for broader adoption to further graph-structured protection in the cryptographic domain.
Implementing new graph models may necessitate customized protocols to accommodate the unique operations and computations introduced by these models.
Beyond GCNs, \cgnn could could facilitate the future exploration of instantiating secure graph models, \eg, GraphSage~\cite{nips/HamiltonYL17}, GAT~\cite{iclr/VelickovicCCRLB18}.

\section{Conclusion}

We propose $\cgnn$, a secure $2$PC framework for GCN inference and training over vertically partitioned data, a neglected MPC scenario motivated by cross-institutional business collaboration.
It is supported by our \osmm protocol using a sparse matrix decomposition method for converting an arbitrary-sparse matrix into a sequence of linear transformations and employing $1$-round MPC protocols of oblivious permutation and selection-multiplication for efficient secure evaluation of these linear transformations.

Our work provides an open-source baseline and extensive benchmarks for practical usage.
Theoretical and empirical analysis demonstrate $\cgnn$'s superior communication and memory efficiency in private GCN computations.
Hopefully, our insight could motivate further research on private graph learning.

\section*{Acknowledgments}

Yu Zheng sincerely appreciates the valuable discussion, editorial helps or comments from  Andes Y.L. Kei, Zhou Li, Sherman S.M. Chow, Sze Yiu Chau, and Yupeng Zhang. 

\bibliography{reference}
\bibliographystyle{ACM-Reference-Format}

\appendix



\section{More Experimental Results}
\label{sec::more_exp}
Tables~\ref{table:time_fea_dim_infer_cora},\ref{table:time_fea_dim_infer_cite},\ref{table:time_fea_dim_infer_pubmed},\ref{table:time_fea_dim_train_cora},\ref{table:time_fea_dim_train_cite},\ref{table:time_fea_dim_train_pubmed} present the inference and training time with varying feature dimensionality over Cora and Pubmed datasets.
The results align with our conclusion in Section~\ref{sec:ablation}.

\begin{table*}[!t]
 	\centering
 	\caption{Inference Time (seconds) with Varying Feature Dimensionality over Cora}
 	\label{table:time_fea_dim_infer_cora}
 \setlength\tabcolsep{8pt}
 	 	\begin{tabular}{c|c|c|c|c|c|c|c|c|c}
 	 	\hline
 \multirow{2}{*}{\textbf{$\#$Dim}} &\multicolumn{3}{c|}{\textbf{$50$Mbps}} &\multicolumn{3}{c|}{\textbf{$100$Mbps}}&\multicolumn{3}{c}{\textbf{$200$Mbps}}
 \\\cline{2-10}
 	& \textbf{Beaver} & ${\prosmm}$ & \textbf{Saving} & \textbf{Beaver} & ${\prosmm}$ 	& \textbf{Saving} & \textbf{Beaver} & ${\prosmm}$ & \textbf{Saving} \\
 	 	\hline
 	 $1433$&$64.18$&$33.53$&$47.8\%$&$35.99$&$20.74$&$ 42.4\%$&$22.03$& $14.62$& $33.6\%$\\
 	 $10$&$51.96$&$19.55$&$62.4\%$&$29.56$&$14.33$&$51.5\%$&$18.89$& $10.80$& $42.8\%$\\
 	 $20$&$52.59$&$19.06$&$63.8\%$&$30.25$&$13.86$&$54.2\%$&$18.48$& $10.62$& $42.5\%$\\
 	 $50$&$52.04$&$19.42$&$62.7\%$&$29.71$&$13.96$&$53.0\%$&$18.69$& $11.38$& $39.1\%$\\\hline
 	 \end{tabular}
 	
 	$\#$E/N: ratio of edges per node,
    ``Beaver'': using Beaver triples for MM.
\end{table*}

\begin{table*}[!t]
 	\centering
 	\caption{Inference Time (seconds) with Varying Feature Dimensionality over Citeseer}
 	\label{table:time_fea_dim_infer_cite}
 \setlength\tabcolsep{8pt}
 	 	\begin{tabular}{c|c|c|c|c|c|c|c|c|c}
 	 	\hline
 \multirow{2}{*}{\textbf{$\#$Dim}} &\multicolumn{3}{c|}{\textbf{$50$Mbps}} &\multicolumn{3}{c|}{\textbf{$100$Mbps}} &\multicolumn{3}{c}{\textbf{$200$Mbps}}
 \\\cline{2-10}
 	 & \textbf{Beaver} & ${\prosmm}$ & \textbf{Saving} & \textbf{Beaver} & ${\prosmm}$ 	 & \textbf{Saving} & \textbf{Beaver} & ${\prosmm}$ & \textbf{Saving}
 	\\\hline
 	$3703$ &$117.83$ &$62.36$ &$47.1\%$ &$64.04$ &$35.78$ &$44.1\%$ &$37.36$ & $23.29$ & $37.7\%$\\
	$10$ &$77.18$ &$20.69$ &$73.2\%$ &$43.22$ &$14.79$ &$65.8\%$ &$26.31$ & $10.66$ & \textbf{59.5\%}\\
 	$20$ &$76.81$ &$20.11$ &\textbf{73.8\%} &$43.59$ &$14.39$ &$67.0\%$ &$25.95$ & $11.06$ & $57.4\%$\\
 	 $50$ &$77.60$ &$20.78$ &$73.2\%$ &$43.69$ &$ 14.31$ &\textbf{67.2\%} &$26.71$ & $11.41$ & $57.3\%$\\\hline
 	 \end{tabular}
\end{table*}

 \begin{table}[!t]
 	\centering
 	\caption{Inference Time (seconds)
  over Pubmed}
 	\label{table:time_fea_dim_infer_pubmed}
 \setlength\tabcolsep{2pt}
 	 	\begin{tabular}{c|c|c|c|c|c|c}
 	 	\hline
 \multirow{2}{*}{\textbf{$\#$Dim}} &\multicolumn{2}{c|}{\textbf{$50$Mbps}} &\multicolumn{2}{c|}{\textbf{$100$Mbps}}&\multicolumn{2}{c}{\textbf{$200$Mbps}}
 \\\cline{2-7}
 	& \textbf{Beaver} & ${\prosmm}$ & \textbf{Beaver} & ${\prosmm}$ 	& \textbf{Beaver} & ${\prosmm}$ \\
 	 	\hline
 	 	 $500$&OOM&$132.06$ &OOM&$69.66$&OOM& $42.37$\\
 	 $10$&OOM&$92.70$&OOM &$53.35$&OOM& $32.77$\\
 	 $20$&OOM&$92.94$&OOM &$53.88$&OOM& $34.22$\\
 	 $50$&OOM&$93.87$&OOM &$54.67$&OOM& $33.49$\\\hline
 	 \end{tabular}
\end{table}

\begin{table*}[!t]
 	\centering
 	\caption{Training Time (seconds) with Varying Feature Dimensionality over Cora}
 	\label{table:time_fea_dim_train_cora}
\setlength\tabcolsep{8pt}
 	 	\begin{tabular}{c|c|c|c|c|c|c|c|c|c}
 	 	\hline
 \multirow{2}{*}{\textbf{$\#$Dim}} &\multicolumn{3}{c|}{\textbf{$50$Mbps}} &\multicolumn{3}{c|}{\textbf{$100$Mbps}}&\multicolumn{3}{c}{\textbf{$200$Mbps}}
 \\\cline{2-10}
 	& \textbf{Beaver} & ${\prosmm}$ & \textbf{Saving} & \textbf{Beaver} & ${\prosmm}$ 	& \textbf{Saving} & \textbf{Beaver} & ${\prosmm}$ & \textbf{Saving} \\
 	 	\hline
 	 	 $1433$&$106.31$&$50.50$&$52.5\%$&$53.19$&$25.79$&$51.5\%$&$22.03$& $13.58$& $38.4\%$\\
 	 $10$&$86.05$&$29.75$&$65.4\%$&$43.05$&$15.62$&$63.7\%$&$18.89$& $8.42$& $55.4\%$\\
 	 $20$&$86.18$&$29.89$&$65.3\%$&$43.13$&$15.47$&$64.1\%$&$18.48$& $8.48$& $54.1\%$\\
 	 $50$&$86.62$&$30.30$&$65.0\%$&$43.34$&$15.77$&$63.6\%$&$18.69$& $8.65$& $53.7\%$\\\hline
 	 \end{tabular}
\end{table*}

\begin{table*}[!t]
 	\centering
 	\caption{Training Time (seconds) with Varying Feature Dimensionality over Citeseer}
 	\label{table:time_fea_dim_train_cite}
 \setlength\tabcolsep{8pt}
 	 	\begin{tabular}{c|c|c|c|c|c|c|c|c|c}
 	 	\hline
 \multirow{2}{*}{\textbf{$\#$Dim}} &\multicolumn{3}{c|}{\textbf{$50$Mbps}} &\multicolumn{3}{c|}{\textbf{$100$Mbps}} &\multicolumn{3}{c}{\textbf{$200$Mbps}}
 \\\cline{2-10}
 	 & \textbf{Beaver} & ${\prosmm}$ & \textbf{Saving} & \textbf{Beaver} & ${\prosmm}$ 	 & \textbf{Saving} & \textbf{Beaver} & ${\prosmm}$ & \textbf{Saving} \\
 	 	\hline
 	 	 $3703$ &$190.04$ &$94.60$ &$50.2\%$ &$95.06$ &$47.86$ &$49.7\%$ &$47.65$ & $24.58$ & $48.4\%$\\
 	 $10$ &$125.78$ &$30.08$ &\textbf{76.1\%} &$62.91$ &$14.79$ &\textbf{76.5\%} &$31.54$ & $8.35$ & \textbf{73.5}\%\\
 	 $20$ &$125.97$ &$30.23$ &$76.0\%$ &$63.03$ &$15.75$ &$75.0\%$ &$31.57$ & $8.46$ & $73.2\%$\\
 	 $50$ &$126.47$ &$30.83$ &$75.6\%$ &$63.29$ &$15.98$ &$74.8\%$ &$31.68$ & $8.58$ & $72.9\%$\\\hline
 	 \end{tabular}
\end{table*}

\begin{table}[!t]
 	\centering
 	\caption{Training Time (seconds) over Pubmed}
 	\label{table:time_fea_dim_train_pubmed}
 \setlength\tabcolsep{3pt}
 	 	\begin{tabular}{c|c|c|c|c|c|c}
 	 	\hline
 \multirow{2}{*}{\textbf{$\#$Dim}} &\multicolumn{2}{c|}{\textbf{$50$Mbps}} &\multicolumn{2}{c|}{\textbf{$100$Mbps}}&\multicolumn{2}{c}{\textbf{$200$Mbps}}
 \\\cline{2-7}
 	& \textbf{Beaver} & ${\prosmm}$ & \textbf{Beaver} & ${\prosmm}$ & \textbf{Beaver} & ${\prosmm}$ \\
 	 	\hline
 	 	 $500$&OOM&$233.53$ &OOM&$118.23$&OOM& $64.16$\\
 	 $10$&OOM&$175.14$&OOM &$92.90$&OOM& $51.27$\\
 	 $20$&OOM&$176.75$&OOM &$92.52$&OOM& $51.72$\\
 	 $50$&OOM&$180.08$&OOM &$93.53$&OOM& $52.71$\\\hline
 	 \end{tabular}
$\#$E/N: edges per node ratio,
``Beaver'': Beaver triples for MM
\end{table}
\begin{table*}[!t]
 	\centering
 	\caption{$10$-Epoch Running Time (seconds) of \osmm in Varying Networks}
 	\label{table:time_smm_net_10}
 \setlength\tabcolsep{7pt}
 	 	\begin{tabular}{c|c|r|r|r|r|r|r|r|r|r}
 	 	\hline
 	\multirow{2}{*}{\textbf{$\#$E/N}} & 
 	\multirow{2}{*}{\textbf{$\#$Node}} & 
 	\multicolumn{3}{c|}{\textbf{Normal ($800$Mbps, $0.022$ms)}} & 
 	\multicolumn{3}{c|}{
 	\textbf{N.B.~($200$Mbps, $0.022$ms)}} & 
 	\multicolumn{3}{c}{
 	\textbf{H.L.~($800$Mbps, $50$ms)}}
 	\\\cline{3-11}
 	 & & \textbf{Beaver} & \osmm & \textbf{Saving} &
 	\textbf{Beaver} & \osmm & \textbf{Saving} &
 	\textbf{Beaver} & \osmm & \textbf{Saving} \\
 	 	\hline
 	$1$ & $1000$ & $3.79$ & $2.92 $ & {$23.0\%$} & $8.89$ &$3.07$ & {$65.5\%$} & $7.55$ &$8.57$ & {$11.9\% $} \\
 	$1$ & $3000$ & $60.67$ & $4.70$ & {$92.3\%$} & $106.39$ & $5.15$ & {$95.2\%$} & $74.54$ & $10.00$ & {$86.6\%$} \\
 	$1$ & $5000$ & $165.35$ & $6.30$ & {$\mathbf{96.2\%}$} & $294.11$ & $7.01$ & {$\mathbf{97.6\%}$} & $203.77$ & $11.53$ & {$\mathbf{94.3\%}$} \\
 	\hline
 	$2$ & $1000$ & $5.95$ & $5.54$ & {$6.9\%$} & $10.92$ & $5.72$ & {$47.6\%$} & $9.17$ & $10.88$ & {$-$} \\
 	$2$ & $3000$ & $62.48$ & $7.99$ & {$87.2\%$} & $108.42$ & $8.57$ & {$92.1\%$} & $76.64$ & $13.24$ & {$82.7\%$} \\
 	$2$ & $5000$ & $168.06$ & $10.57$ & {$\mathbf{93.7\%}$} & $296.50$ & $11.71$ & {$\mathbf{96.1\%}$} & $206.38$ & $15.82$ & {$\mathbf{92.3\%}$} \\
 	\hline
 	$3$ & $1000$ & $8.27$ & $8.62$ & {$4.2\%$} & $13.61$ & $8.89$ & {$34.7\%$} & $11.83$ & $13.93$ & {$-$} \\
 	$3$ & $3000$ & $64.45$ & $11.92$ & {$81.5\%$} & $111.71$ & $12.74$ & {$88.6\%$} & $79.58$ & $17.19$ & {$78.4\%$} \\
 	$3$ & $5000$ & $170.01$ & $15.28$ & {$\mathbf{91.0\%}$} & $299.57$ & $16.77$ & {$\mathbf{94.4\%}$} & $209.03$ & $20.64$ & {$\mathbf{90.1\%}$} \\
 	\hline
 	 	\end{tabular}
        
 	 $\#$E/N: ratio of edges per node,
	``Beaver'': using Beaver triples
\end{table*}


\section{Proofs related to Sparsity}
\label{sec:matrix_found_sparse}

\begin{definition}[$\Qsf$-type matrix]
\label{q_type}
A $(0,1)$-matrix $\Msf$ of size $m\times n$ is a $\Qsf$-type matrix iff there exists a monotonically non-decreasing 
$f:\Zbb/m\Zbb\rightarrow \Zbb/n\Zbb$ s.t.~$\Msf[i,j]=1$ iff $j=f(i)$.
\end{definition}

\begin{definition}[$\Psf$-type matrix] 
\label{p_type}
A $(0,1)$-matrix $\Msf$ of size $m\times n$ is a $\Psf$-type matrix iff there exists a monotonically non-decreasing 
$f:\Zbb/n\Zbb\rightarrow \Zbb/m\Zbb$ s.t.~$\Msf[i,j]=1$ iff $i=f(j)$.
\end{definition}

\subsection{Proof of Theorem~\ref{the::p_sig_q}}
\label{sec::proof_p_sig_q}
\begin{proof}
Suppose the sparse representation of $\adjmat$ is $\{(i_k, j_k, \lambda_k): k=1, \ldots, t\}$ where
$k \mapsto i_k$ is monotonically non-decreasing.
Then, we have $\adjmat=\adjout'\Lambda \adjin$, where $\adjout'\in \Mbb_{m,t}(\Rcal)$ is a sparse matrix represented by $\{(i_k, k, 1): k=1, \ldots, t\}$, $\Lambda=\mathsf{diag}(\lambda _1, \ldots, \lambda _t)$, $\adjin\in \Mbb_{t,n}(\Rcal)$ is a 
sparse matrix represented by $\{(k, j_k, 1): k=1, \ldots, t\}$.
Now, $\adjout'$ is a $\Psf$-type matrix, but $\adjin$ is not a $\Qsf$-type matrix as $\adjin$ does not satisfies ``$k \mapsto j_k$ is monotonically non-decreasing''.
We permute the lines of $\adjin$ as $\adjin'=\{ (\sigma_3^{-1}(k), j_k, 1): k=1, \ldots, t\}=\{ (k, j_{\sigma_3(k)}, 1): k=1\ldots, t\}$ such that $k \mapsto j_{\sigma _3(k)}$ is monotonically non-decreasing.
After that, we have $\adjin= \{ (k, \sigma_3^{-1}(k), 1): k=1, \ldots, t \} \times \adjin' = \{ (\sigma_3(k), k, 1): k=1, \ldots, t \} \times \adjin' =\sigma_3 \adjin'$, where $\adjin'$ is a $\Qsf$-type matrix.
Hence, $\adjmat=\adjout' \Lambda \sigma _3 \adjin'$.
\end{proof}

\subsection{Proof of Theorem~\ref{the::general_dec_main}} 
\label{sec::proof_general}

\begin{proof}
The proof is straightforward by composing $\adjmat = \adjout' \Lambda \sigma_3 \adjin'$ from Theorem \ref{the::p_sig_q}, $\adjin'=\Sigma \sigma_2 \gammain \delta _n \sigma _1$ from Theorem~\ref{the::q_decom}, and $\adjout'=\sigma_5 \delta_m ^\top \gammaout \sigma_4 \Sigma ^\top$ from Theorem~\ref{the::p_decom}.
\end{proof}

\subsection{Proof of Theorem~\ref{the::smm_main}}
\label{sec::proof_smm}
\begin{proof}
It is straightforward to prove by observing that:\\
1) Permutations on $\sigma _i$ calls permutation group action;\\
2) Multiplying 
$\delta$, $\Sigma$ calls constant matrix multiplication;\\
3) Multiplying $\gammain$, $\gammaout$ calls element-wise multiplication (with cut-off and padding of zero values);\\
4) Multiplying $\Lambda$ calls element-wise multiplication.	
\end{proof}

\subsection{Proof of Theorem~\ref{the::q_decom}}

\label{sec::proof_q}

\begin{theorem}
\label{the::q_decom}
Let $\adjin' \in \Mbb_{t, n}$ be a $\Qsf$-type matrix with $\ncol$ non-zero columns.
Then, there exists a matrix decomposition $\adjin'=\Sigma \sigma_2 \gammain \delta _n \sigma _1$ where $\sigma _1 \in \Sbb_n, \sigma _2 \in \Sbb_t$, and, \\
1) $\Sigma=(\Sigma[i, j])_{i,j=1}^{\numnonzero}$ is the left-down triangle matrix such that $\Sigma[i, j]=1$ if $i \geq j$ or $0$ otherwise,\\
2) $\delta_n=(\delta_n[i,j])_{i,j=1}^{n}$ is the left-down triangle matrix such that $\delta_n[i,j]=1$ for $i=j$ or $-1$ for $j=i-1$, or $0$ otherwise,\\
3) $\gammain =(\gammain[i,j])_{i=1,j=1}^{\numnonzero,n}$ is a matrix such that $\gammain[i,j]=1$ for $1\leq i=j\leq \ncol$ or $0$ otherwise.
\end{theorem}
 
\begin{proof}
Here, we prove that $\adjin' \Xsf=\Sigma \sigma _2 \gammain \delta_n \sigma _1 \Xsf $ holds for any $\Xsf\in \Rcal^{(n)}$.
Firstly, we can use column transformation to transform matrix $\adjin'$ to a new matrix $\tilde{\adjin'}$ of $\Qsf$-type such that all the $n-\ncol$ zero-columns of $\tilde{\adjin'}$ lie in the last columns.
Hence, we have $\adjin'=\tilde{\adjin'}\sigma _1$, and $\tilde{\adjin'}$ is in the form of,
\[
\tilde{\adjin'}=\left(
\begin{array}{ccccccc}
1 &&& & 0 & \cdots & 0\\
\vdots &&& \\
1 &&& & \vdots & & \vdots \\
 & 1 && \\
 & \vdots && & \vdots & & \vdots \\
 & 1 && \\
 & & \ddots & & \vdots & & \vdots \\
 && & 1 \\
 && & \vdots \\
 &&& 1 & 0 & \cdots & 0\\
\end{array}
\right)
\]
Let $\tilde{\Xsf} = \sigma_1 \Xsf$,
then we have $\adjin' \Xsf$ equal to:
\begin{equation*}
	\begin{aligned}
		\tilde{\adjin'}\tilde{\Xsf}&=
\left(
\begin{array}{c}
\tilde{x}_1 \\
\tilde{x}_1 \\
\vdots \\
\tilde{x}_1 \\
\tilde{x}_2 \\
\tilde{x}_2 \\
\vdots \\
\tilde{x}_2 \\
\vdots \\
\vdots \\
\tilde{x}_{\ncol} \\
\tilde{x}_{\ncol} \\
\vdots \\
\tilde{x}_{\ncol} \\
\end{array}
\right) = \Sigma
\left(
\begin{array}{c}
\tilde{x}_1 \\
0 \\
\vdots \\
0 \\
\tilde{x}_2-\tilde{x}_1 \\
0 \\
\vdots \\
0 \\
\vdots \\
\vdots \\
\tilde{x}_{\ncol}-\tilde{x}_{\ncol-1} \\
0 \\
\vdots \\
0 \\
\end{array}
\right)
\end{aligned}
\end{equation*}
\begin{equation*}
\begin{aligned}
& =
\Sigma \sigma _2
\left(
\begin{array}{c}
\tilde{x}_1 \\
\tilde{x}_2-\tilde{x}_1 \\
\vdots \\
\tilde{x}_{\ncol}-\tilde{x}_{\ncol-1} \\
0 \\
\vdots \\
0 \\
\end{array}
\right)
 = \Sigma \sigma _2 \gammain
\left(
\begin{array}{c}
\tilde{x}_1 \\
\tilde{x}_2-\tilde{x}_1 \\
\vdots \\
\tilde{x}_{\ncol}-\tilde{x}_{\ncol-1} \\
\tilde{x}_{\ncol+1}-\tilde{x}_{\ncol} \\
\vdots \\
\tilde{x}_{n}-\tilde{x}_{n-1} \\
\end{array}
\right)
\\
&
= \Sigma \sigma _2 \gammain \delta_n \tilde{\Xsf}
= \Sigma \sigma _2 \gammain \delta_n \sigma _1 \Xsf
\rlap{\qquad \qquad \qquad \qquad \qquad \qquad \qedhere}.
\end{aligned}
\end{equation*}
\end{proof}

\subsection{Proof of Theorem~\ref{the::p_decom}}

\label{sec::proof_p}
\begin{theorem} 
\label{the::p_decom}
Let $\adjout' \in \Mbb_{m,t}(\Rcal)$ be a $\Psf$-type matrix with $\nrow$ non-zero rows.
Then, there exists a matrix decomposition $\adjout'=\sigma_5 \delta_m ^\top \gammaout \sigma_4 \Sigma ^\top$ where $\sigma _5 \in \Sbb_m$, $\sigma_4 \in \Sbb_t$, and, \\
1) $\Sigma=(\Sigma[i, j])_{i,j=1}^{\numnonzero}$ is the left-down triangle matrix such that $\Sigma[i, j]=1$ if $i \geq j$ or $0$ otherwise,\\
2) $\delta_m=(\delta_m[i,j])_{i,j=1}^{m}$ is the left-down triangle matrix such that $\delta_m[i,j]=1$ for $i=j$ or $-1$ for $j=i-1$, or $0$ otherwise,\\
3) $\gammaout =(\gammaout[i,j])_{i=1,j=1}^{m,\numnonzero}$ is a matrix such that $\gammaout[i,j]=1$ for $1\leq i=j\leq {\nrow}$ or $0$ otherwise.
\end{theorem}

\begin{proof}
Note that $(\adjout')^\top \in \Mbb_{t,m}$ is a $\Qsf$-type matrix, we have the matrix decomposition $(\adjout')^\top=\Sigma \sigma_2 \gammain \delta_m \sigma _1$ by Theorem~\ref{the::q_decom}, where $\sigma_1 \in \Sbb_m$ and $\sigma_2\in \Sbb_t$.
Let $\sigma_5=\sigma_1^\top$, $\sigma_4=\sigma_2^\top$, $\gammaout=\gammain^\top$, we have $\adjout'=\sigma_5 \delta_m ^\top \gammaout \sigma_4 \Sigma ^\top$.
\end{proof}

 
\section{Algorithm Realization}
\label{sec::alg}

Guided by Theorem~\ref{the::p_sig_q}, our algorithm of decomposing the adjacency matrix $\Asf$ can be implemented as below.
$\Asf$ is a \texttt{SparseMatrix} consists of all $0$-$1$ elements represented by coordinates of nonzero values,
\ie, [(\texttt{row-index}, \texttt{column-index}), (\texttt{row-index}, \texttt{column-index}), $\ldots$].
In Figure~\ref{fig::code_psigq_de}, we draw the pairs of (row-index, column-index) of $\Asf$, \eg, $[(2, 3), (3, 1), \ldots]$.
Step~1 is to split the pairs in $\Asf$ to construct two sparse matrices (Graph language in Figure~\ref{fig::nen_relation_diff}), called $\Psf'$ and $\Qsf'$ (correspond to $\adjout$ and $\adjin$ in Section~\ref{subsec::sme}).
In Steps~2 and~3, the row-indices of $\Psf'$ are sorted and generate $\sigma_3^{\Psf}$, and the column-indices of $\Qsf'$ are sorted and generate $\sigma_3^{\Qsf}$.
Then, $\sigma_3$ is obtained by sparse-matrix multiplying $\sigma_3^{\Psf}$ and $\sigma_3^{\Qsf}$, thus $\sigma_3=\sigma_3^{\Psf}\cdot\sigma_3^{\Qsf}$.
Now, we can see that the resulting matrices $\Psf$ and $\Qsf$ (correspond to $\adjout'$ and $\adjin'$ in Section~\ref{subsec::initdecom}) have monotonically non-decreasing row/column indices.
Moreover, $\Psf$ contain exactly one $1$ in every column and $\Qsf$ contains exactly one $1$ in every row.
In summary, the pseudocode of \texttt{decompose\_row\_column(A)} below describes the extraction of $\Psf$, $\sigma_3$ and $\Qsf$ as previously displayed in Figure~\ref{graph_psigq}.


\begin{figure}[!t]
	\centering
	\includegraphics[width = 0.47\textwidth]{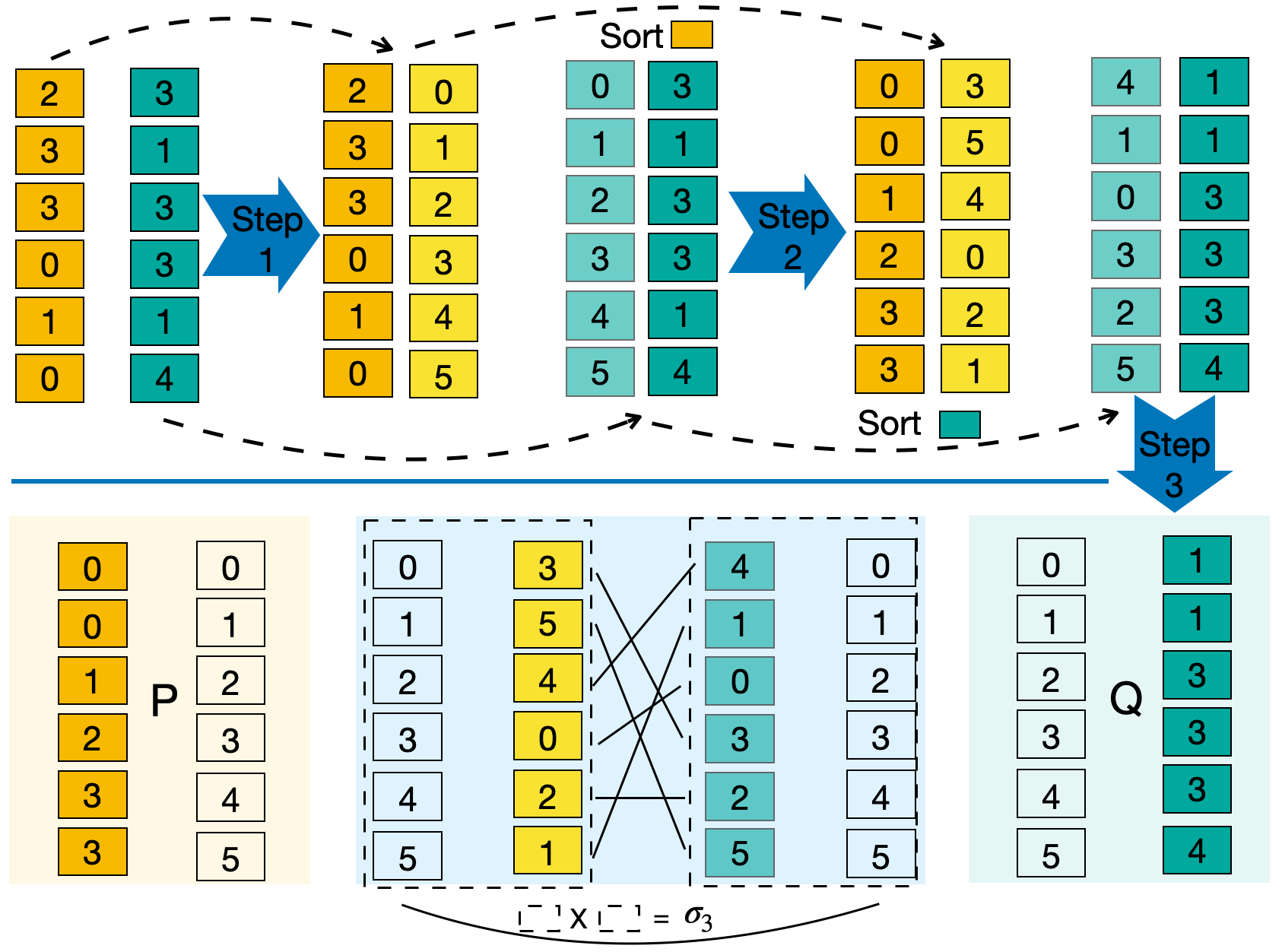}
	\caption{Illustration of $\Psf\sigma_3\Qsf$ Decomposition}
	 \label{fig::code_psigq_de}
\end{figure}

\noindent
\begin{frame}{\texttt{def decompose\_row\_column(A):}}
 \\\texttt{\#A: a list of (row\_id, col\_id) for non-zero element in sparse matrix}
\\\indent\texttt{P=[(A[i, 0], i) for i in range(len(A))]}
\\\indent\texttt{Q=[(i, A[i, 1]) for i in range(len(A))]}
\\\indent\texttt{P=sorted(P, key=lambda x: x[0])}
\\\indent\texttt{Q=sorted(Q, key=lambda x: x[1])}
\\\indent\texttt{P=[(P[i,0], i) for i in range(len(A))]}
\\\indent\texttt{sigma3P=[(i, P[i,1]) for i in range(len(A))]}
\\\indent\texttt{sigma3Q=[(Q[i,0], i) for i in range(len(A))]}
\\\indent\texttt{Q=[(i,Q[i,1]) for i in range(len(A))] }
\\\indent\texttt{sigma3 = sigma3P*sigma3Q}
\\\indent\texttt{return P, sigma3, Q}
\end{frame}

\noindent\textbf{Pseudocode of re-decomposition}.
In Figure~\ref{fig::code_q_compos}, we re-draw $\Qsf$ and describe its decomposition.
Step~1 extracts the unique column indices using $\text{set}$ function in Python, and their quantity is $\ncol$.
Then, the corresponding row indices are extracted by comparing whether the neighboring column indices are identical.
Step~2 constructs $\sigma_1$ and $\sigma_2$ by keeping the first $\ncol$ elements and padding the elements in numerical order to a permutation in $\sigma_1\in\Sbb_n,\sigma_2\in\Sbb_t$.
The code of decomposing $\Qsf$ is outlined below.
To derive $\sigma_4$, $\nrow$, and $\sigma_5$, the $\Psf$-type matrix decomposition follows a similar logic.

\begin{figure}[!t]
	\centering
	\includegraphics[width = 0.42\textwidth]{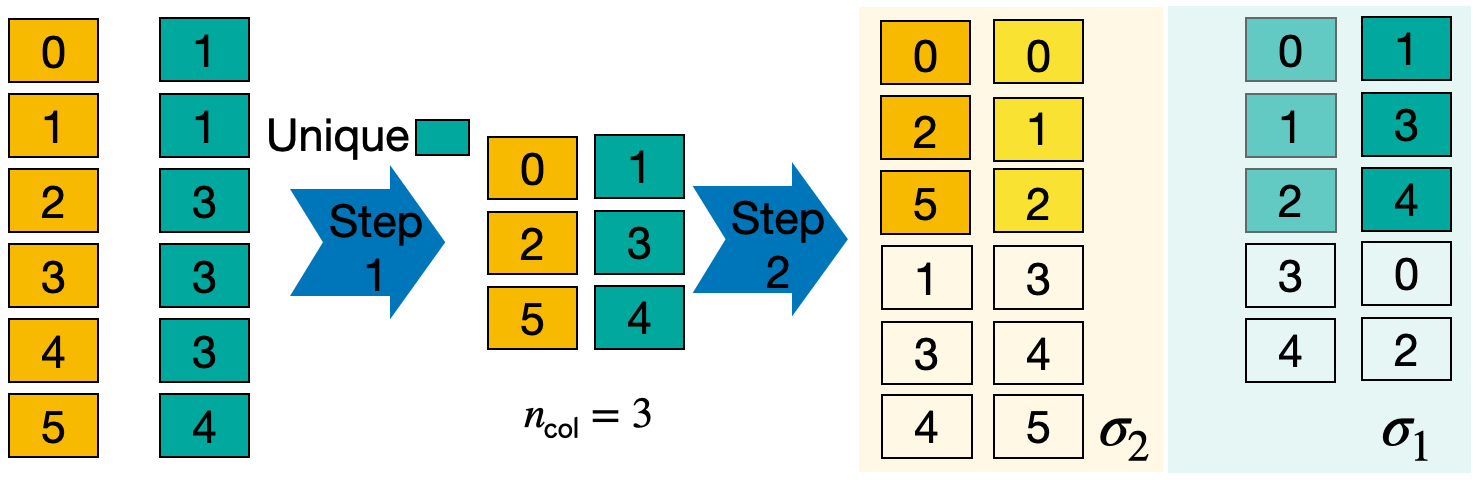}
	\caption{Illustration of $\Qsf$ Decomposition}
	 \label{fig::code_q_compos}
\end{figure}

\begin{frame}{\texttt{def decompose\_Q(Q, e, n):}}
 \\\indent\texttt{unique\_col\_ids = set(Q[:,1])} 
 \\\indent\texttt{step\_row\_ids=[Q[i,0] for Q[i, 1)!=Q[i-1, 1) or i=0]}
 \\\indent\texttt{k2 = len(unique\_col\_ids)}
 \\\indent\texttt{sigma1 = [ (i, unique\_col\_ids[i]	) for i in range(k2)]}
 \\\indent\texttt{sigma2 = [(step\_row\_ids[i], i) for i in range(k2)] }
\\\indent\texttt{sigma1=pad\_perm(sigma1, n)} 
\\\indent\texttt{sigma2=pad\_perm(sigma2, e)} 
\\\indent\texttt{return sigma2, k2, sigma1}
\end{frame}

The \texttt{class PrivateSparseMatrix} below contains the realization of \osmm protocols.
Before secure training, the graph owner locally decomposes $\Asf$ into $\Psf$, $\Qsf$, and then into the corresponding basic operations.
\osmm let $\pp_0$ and $\pp_1$ jointly execute secure multiplications on $\Psf$ and $\Qsf$, and $\ssp$ on $\sigma_3$.

\vspace{2mm}
\noindent
\begin{frame}{\texttt{class PrivateSparseMatrix:}}
\\\indent\texttt{def \_\_init\_\_(self, A ...):}
\\\indent\indent\texttt{self.owner = get\_device()}
\\\indent\indent\texttt{with tf.device(self.owner):}
\\\indent\indent\indent\texttt{P, s3, Q = decompose\_row\_column()}
\\\indent\indent\indent\texttt{s5, k4, s4 = decompose\_P()}
\\\indent\indent\indent\texttt{s2, k2, s1 = decompose\_Q()}
\\\indent\texttt{def sm\_2(self, x: Union[PrivateTensor,} \\\indent\indent\indent\indent\ \indent\indent\indent\indent\texttt{SharedPair]):}
\\\indent\indent\texttt{Qx = Q\_mult(s2, k2, s1, e, x)}
\\\indent\indent\texttt{x3 = s3.act(Qx)}
\\\indent\indent\texttt{Ax = P\_mult(s5, k4, s4, x3, m)}
\\\indent\indent\texttt{return Ax}

\end{frame}

{Class of Secure GCN.}
The implementation of \cgnn follows the plaintext-training repository (\url{https://github.com/dmlc/dgl}) in the transductive setting.
Accordingly, the graph decomposition can be performed once with the fixed graph before secure training.
The \texttt{class SGCN} inherits the conventional \texttt{NN} (the template of neural network).
Secure GCN training is thus composed of secure graph convolution and secure activation layers.
The $\Asf\Xsf$ in graph convolution layers is realized by the $\Pi_\smm$ protocol.
Below, we extract the code of implementing the \texttt{class SGCN} with respect to the plaintext GCN.
Notably, all the inputs \texttt{feature}, \texttt{label}, and \texttt{adj\_matrix} are secret shares in the form of fixed-pointed numbers.
The functions \texttt{GraphConv, ReLU, SoftmaxCE} are the MPC protocols executed by two parties.

\begin{frame}{\texttt{class SGCN(NN):}}
\\\texttt{def \_\_init\_\_(self, feature: PrivateTensor, 
\\\indent\indent\indent\indent\indent\indent label: PrivateTensor, 
\\\indent\indent\indent\indent\indent\indent dense\_dims: List[int],
\\\indent\indent\indent\indent\indent\indent adj\_matrix: Union[...],
\\\indent\indent\indent\indent\indent\indent train\_mask, loss=...):}
\\\indent\texttt{super(SGCN, self).\_\_init\_\_()}
\\\indent\texttt{layer = Input(dim, feature)}
\\\indent\texttt{self.addLayer(layer)}
\\\indent\texttt{input\_layers = [layer]}
\\\indent\texttt{for i in range(1, len(dense\_dims)):}
\\\indent\indent\texttt{layer = GraphConv()}
\\\indent\indent\texttt{self.addLayer(ly=layer)}
\\\indent\indent\texttt{if i < len(dense\_dims) - 1:}
\\\indent\indent\indent\texttt{layer = ReLU()} 
\\\indent\indent\indent\texttt{self.addLayer(ly=layer)}
\\\indent\texttt{layer\_label = Input(dim, label)}
\\\indent\texttt{self.addLayer(layer\_label)}
\\\indent\texttt{if loss == "SoftmaxCE":} 
\\\indent\indent\texttt{layer\_loss = SoftmaxCE()}
\\\indent\indent\texttt{self.addLayer(ly=layer\_loss)}
\\\indent\texttt{else:}
\\\indent\indent\texttt{...\# use other layer/loss}
\end{frame}

\section{Selection-Multiplication's Correctness}
\label{sec:proof_lemma}
\begin{proof}[Proof of Lemma~\ref{lem::sxb}]
We analyze the two cases of $s=0$ and $s=1$ for a complete proof where $s \in \{0, 1\}$.

(i). When $s=0$, the equivalence of the first property turns to be $f(0,x+u)=f(0,x)+f(0,u) \Leftrightarrow 0\cdot(x+u) = 0\cdot x+ 0\cdot u \Leftrightarrow 0=0$.
When $s=1$, we have $f(1,x+u)=f(1,x)+f(1,u) \Leftrightarrow 1\cdot(x+u) = 1\cdot x+ 1\cdot u \Leftrightarrow x+u=x+u$.

(ii). When $s=0$, the 
second property 
becomes $f(0+b,x)=f(0,x)+(-1)^0f(b,x) \Leftrightarrow bx = 0\cdot x+ 1\cdot bx \Leftrightarrow bx = bx$.
When $s=1$, we have $f(1+b,x)=f(1,x)+(-1)^1f(b,x) \Leftrightarrow (1+b)\cdot x = 1\cdot x - bx$.
If $b=0$, we get $x=x$.
If $b=1$, we get $0=0$ since $(1+b)\!\!\!\!\mod 2=0$.
\end{proof}
\section{Security Analysis}
\label{app::fullproof}
We prove the semi-honest security of our protocols under the real/ideal-world simulation paradigm~\cite{sp/17/Lindell17} with a hybrid argument.
As our protocols satisfy the stand-alone model without malicious assumption, we adopt the standard simulation proof technique instead of the UC framework that adds an additional ``environment" representing an interactive distinguisher.

We consider the $2$PC executed by $\pp_0$ and $\pp_1$ in the presence of static semi-honest adversaries $\A$ that control one of the parties at the beginning, follow the protocol specification, and try to learn information about the honest party's private input.
Definition~\ref{def:semidef} states the semi-honest security so that simulated and real execution are computationally indistinguishable ("$\equiv$") for~$\A$.
That is, the simulator $\Scal$ can generate the view of a party in the execution, implying the party learns nothing beyond what they can derive from their input and prescribed output.
For simplicity, we assume $\mathsf{PRF}$ to be secure and exclude its standard proof here.

\begin{restatable}[Semi-honest Security~\cite{sp/17/Lindell17}]{definition}{semidef}
\label{def:semidef} 
	Let $\lambda$ be a security parameter.
	A protocol $\Pi$ securely realizes a functionality $\F = (\F_0, \F_1)$ on input $\l x \r = (\l x\r_0, \l x\r_1)$ against 
	static semi-honest adversaries if there exist PPT simulators $\Scal_0, \Scal_1$ s.t.
 \begin{equation*}
     \begin{aligned}
         \{\Scal_0(1^\lambda, \l x\r_0, \F_0), \F(\l x \r)\} \equiv \{\view_0^\Pi, \output^{\Pi, \lambda}(\l x \r)\},\\
         \{\Scal_1(1^\lambda,\l x\r_1,\F_1), \F(\l x \r)\} \equiv \{\view_1^\Pi,\output^{\Pi,\lambda}(\l x \r)\}.
     \end{aligned}
 \end{equation*}
\end{restatable}

\subsection{Security of \texorpdfstring{$\Pi_\ssp$}{πOP}}
\label{sec::detail_op}

We divide the analyses into `raw' and `shared' cases.
Functionality~\ref{func::ssp} presents the ideal functionality $\F_\ssp$.
It contains two cases in which the input vector/matrix $\Xsf$ is owned by one party or secret-shared among two parties.
The functionality $\F_\ssp$ of both cases outputs additive shares of permutation over~$\Xsf$, \ie, $\l \sigma\xvec\r_0 +\l \sigma\xvec\r_1 =\sigma\xvec$.

\setcounter{algorithm}{0}

\begin{functionality}[!t]
	\caption{$\F_\ssp$: Ideal Functionality of $\Pi_{\ssp}$}\label{func::ssp} 
\begin{algorithmic}[1]
	\item[\textbf{Parameter:} Type of input $\type \in \{\raw, \shared\}$.]
	\REQUIRE $\sigma\in\Sbb,\xvec\in \Zbb^m$ if $\type == \raw$; \\ ~~~~~otherwise $\sigma\in \Sbb,\l\xvec\r_0\in\Zbb^m, \l\xvec\r_1\in\Zbb^m$.
	\ENSURE $\l \sigma\xvec\r_0,\l \sigma\xvec\r_1$.
	\IF{$\type == \shared$} 
	\STATE Reconstruct $\Xsf = \l \xvec \r_0+\l \xvec \r_1$
	\ENDIF
	\STATE Compute and generate random shares of $\sigma\xvec$
	\RETURN $\l \sigma\xvec \r$
\end{algorithmic}
\end{functionality}

\begin{theorem}
\label{the::ssp}
	The protocol $\Pi_\ssp$ securely realizes the ideal functionality $\F_\ssp$ against static semi-honest adversaries.
\end{theorem}
 
\begin{proof}
	We define the following $\ideal$ and $\real$ experiments:
	\begin{equation*}
			\begin{aligned}
			\real^{1^\lambda,\A}_{\Pi_\ssp}=&\ \{\{(\view_0^{\Pi,\raw},\l\sigma\Xsf\r_0), (\view_1^{\Pi,\raw},\l\sigma\Xsf\r_1)\} \text{ or }\\
			&\ \ \ \{(\view_0^{\Pi,\shared},\l\sigma\Xsf\r_0), (\view_1^{\Pi,\shared},\l\sigma\Xsf\r_1)\}
			\}\\
			\ideal^{1^\lambda, \A}_{\Scal, \F}=&\ \{\{\Scal(\raw,1^\lambda,\sigma,\Xsf,\F_\ssp),\F_\ssp(\sigma,\Xsf)\}\text{ or }\\
			&\ \ \ \{\Scal(\shared,1^\lambda,\sigma,\l\Xsf\r_0,\l\Xsf\r_1,\F_\ssp),
			\\&\quad\F_\ssp(\sigma,\l\Xsf\r_0,\l\Xsf\r_1)\}\}
		\end{aligned}
	\end{equation*}
	where $\pp_0$'s view is either $\view_0^{\Pi,\raw}$, which is
	$(\sigma, \pi, \l \pi\uvec\r_0, \delta_{\sigma}, \delta_{\Xsf})$ or $\view_0^{\Pi,\shared}$,
	which is $(\sigma, \l \Xsf\r_0, \pi, \l \pi\uvec \r_0, \delta_{\sigma}, \delta_{\l\Xsf\r_1})$,
	and $\pp_1$'s view is either $\view_1^{\Pi,\raw}=(\Xsf, \l \pi\uvec \r_1,\delta_{\sigma},\delta_{\Xsf})$ or $\view_1^{\Pi,\shared}=(\l\Xsf\r_1, \l \pi\uvec \r_1,\delta_{\sigma},\allowbreak\delta_{\l\Xsf\r_1})$.
	$\real^{1^\lambda,\A}_{\Pi_\ssp}$ represents real protocol execution.
	In the $\ideal$ world, the simulators $\Scal=\{\Scal_0,\Scal_1\}$ can indistinguishably simulate the view of each honest party in the protocol given only that party's input.
 
	Now, we argue that $\real^{1^\lambda,\A}_{\Pi_\ssp}\equiv\ideal^{1^\lambda, \A}_{\Scal, \F}$ for any PPT $\A$ using the multi-step hybrid-argument technique.
	 \begin{itemize}
		 \item[$\hyb_0$:] It is identical to the real protocol execution $\real^{1^\lambda,\A}_{\Pi_\ssp}$.
		 \item[$\hyb_{1}$:] It is identical to $\hyb_0$ except that $\delta_\sigma,\delta_{\Xsf}$ are randomly generated for the case of $\raw$ and $\delta_\sigma,\delta_{\l\Xsf\r_1}$ are randomly generated for the case of $\shared$.
		 \\
		 i) In the first case that $\Xsf$'s type is $\raw$, any PPT $\A$ cannot distinguish $\delta_\sigma,\delta_{\Xsf}$ in $\real^{1^\lambda,\A}_{\Pi_\ssp}$ experiment and $\Tilde{\delta}_\sigma,\Tilde{\delta}_{\Xsf}$ in $\hyb_1$ since $\delta_\sigma,\delta_{\Xsf}$ are computed by $\pi,\uvec$, which are generated by $\prf$.
		 If $\A$ can distinguish $\hyb_{1}$ and $\real^{1^\lambda,\A}_{\Pi_\ssp}$ with non-negligible advantage, $\A$ can break the security of $\prf$, which contradicts the assumption.
		 \\
		 ii) For the second case that $\Xsf$'s type is $\shared$, any PPT $\A$ cannot distinguish $\delta_\sigma,\delta_{\l\Xsf\r_1}$ in $\real^{1^\lambda,\A}_{\Pi_\ssp}$ experiment and $\Tilde{\delta}_\sigma,\Tilde{\delta}_{\l\Xsf\r_1}$ in $\hyb_1$ since $\delta_\sigma,\delta_{\l\Xsf\r_1}$ are computed by $\pi,\uvec$, which are generated by $\prf$.
		 If $\A$ can distinguish $\hyb_{1}$ and $\real^{1^\lambda,\A}_{\Pi_\ssp}$ with non-negligible advantage, $\A$ can break the $\prf$ security, contradicting the assumption.
		 \\
		 $\Rightarrow \hyb_1\equiv\hyb_0$.
		 \item[$\hyb_{2}$:] It is identical to $\ideal^{1^\lambda, \A}_{\Scal, \F}$, \ie, all $\view_0,\view_1$ of two parties are simulated by $\Scal_0,\Scal_1$.
		 The randomness of $\pi,\l \pi\uvec \r_0, \delta_{\sigma}, \delta_{\Xsf}$ for $\raw$ and $\pi,\l \pi\uvec \r_0, \delta_{\sigma}, \delta_{\l\Xsf\r_1}$ for $\shared$ ensures no non-negligible $\A$'s advantage of distinguishability to $\Scal_0$'s view.
		 Similarly, the randomness of $\l \pi\uvec \r_1, \delta_{\sigma}, \delta_{\Xsf}$ for $\raw$ and $\l \pi\uvec \r_1, \allowbreak\delta_{\sigma}, \delta_{\l\Xsf\r_1}$ for $\shared$ ensures no non-negligible $\A$'s advantage of distinguishability to $\Scal_1$'s view.
		 Now, $\pp_0$ cannot obtain $\pp_1$ inputs, while $\pp_1$ cannot obtain $\pp_0$ inputs since $\{\Scal_i\}_{i\in\{0,1\}}$ cannot obtain $\{\Scal_{1-i}\}_{i\in\{0,1\}}$'s inputs using the $\{\view_i\}_{i\in\{0,1\}}$.\\
		 $\Rightarrow \hyb_2\equiv\hyb_1$.
	 \end{itemize}
Thus, for both cases, 
we have $\hyb_2\equiv\hyb_1\equiv\hyb_0$, equivalent to $\real^{1^\lambda,\A}_{\Pi_\ssp}\equiv\ideal^{1^\lambda, \A}_{\Scal, \F}$ for any PPT semi-honest $\A$.
\end{proof}

\subsection{Security of 
\texorpdfstring{$\Pi_\SM$}{πOSM}
} 
\label{sec::proof_sxb}


Functionality~\ref{func::osm} gives the ideal functionality of $\Pi_{\SM}$, defining the multiplication $sx\in\Zbb_{2^n}$ between 
$s\in\Zbb_2$ and 
$x\in\Zbb_{2^n}$.

\begin{functionality}[!t]
	\caption{$\F_{\SM}$: Ideal Functionality of $\Pi_{\SM}$}\label{func::osm}
\begin{algorithmic}[1]
	\REQUIRE $s\in\Zbb_2, \l x\r_0\in\Zbb_{2^n}, \l x\r_1\in\Zbb_{2^n}$.
	\ENSURE $\l sx \r_0, \l sx \r_1$.
	\STATE Reconstruct $s=\l s\r_0+\l s\r_1$
	\STATE Compute and generate random shares of $sx$
	\RETURN $\l sx \r$
\end{algorithmic}
\end{functionality}

\begin{theorem}
\label{the::osm}
	The protocol $\Pi_\SM$ securely realizes the ideal functionality $\F_\SM$ against static semi-honest adversaries.
\end{theorem}

\begin{proof}
We define the $\ideal$ and $\real$ experiments:
	\begin{equation*}
		\begin{aligned}
			\real^{1^\lambda,\A}_{\Pi_\SM}=&\ \{(\view_0^{\Pi},\l sx\r_0), (\view_1^{\Pi},\l sx\r_1)\} \\
			\ideal^{1^\lambda, \A}_{\Scal, \F}=&\ \{\Scal(1^\lambda,s,\l x\r_0,\l x\r_1,\F_\SM),\\&\ \ \ \ \F_\SM(s,\l x\r_0,\l x\r_1)\}
		\end{aligned}
	\end{equation*}
	where 
	$\view_0^{\Pi}=(s, \l x\r_0, b, \l u \r_0, \l bu \r_0,\delta_s, \allowbreak\delta_{\l x\r_0},\delta_{\l x\r_1},\delta_x)$, and 
	$\view_1^{\Pi}=(\l x\r_1, \l bu \r_1, \l u \r_1, \delta_{\l x\r_1}, \delta_s)$.
	$\real^{1^\lambda,\A}_{\Pi_\SM}$ represents real protocol execution.
	The simulators $\Scal=\{\Scal_0,\Scal_1\}$ are indistinguishably simulating the view of each honest party in the protocol given only that party's input.
 
	Now, we prove that $\real^{1^\lambda,\A}_{\Pi_\SM}\equiv\ideal^{1^\lambda, \A}_{\Scal, \F}$ for any PPT $\A$ with a series of hybrid-arguments.
	The hybrid games can be sequentially formulated as follows.
\begin{itemize}
		 \item[$\hyb_0$:] It is identical to the real protocol execution $\real^{1^\lambda,\A}_{\Pi_\SM}$.
		 \item[$\hyb_{1}$:] It is identical to $\hyb_0$ except that $\delta_{\l x\r_0},\delta_{\l x\r_1},\delta_{x}$ are randomly generated 
		 by $\Scal_0$.
		 Since $\delta_{\l x\r_0},\delta_{\l x\r_1}$ in $\real^{1^\lambda,\A}_{\Pi_\SM}$ experiment are computed by $\l u \r_0,\l u \r_1$, which are outputted by $\prf$.
		 Thus, any PPT $\A$ cannot distinguish $\delta_{\l x\r_0},\delta_{\l x\r_1}$ in $\real^{1^\lambda,\A}_{\Pi_\SM}$ experiment and $\Tilde{\delta}_{\l x\r_0},\Tilde{\delta}_{\l x\r_1}$ in $\hyb_1$, guaranteed by $\prf$'s security.
		 The value of $\delta_x$, added by $\delta_{\l x\r_0}$ and $\delta_{\l x\r_1}$, is also distinguishable to $\Tilde{\delta}_x$ simulated by $\Scal_0$.
		 Overall, if $\A$ can distinguish $\Tilde{\delta}_{\l x\r_0},\Tilde{\delta}_{\l x\r_1},\Tilde{\delta}_x$ with $\delta_{\l x\r_0},\delta_{\l x\r_1},\delta_x$ in $\real^{1^\lambda,\A}_{\Pi_\SM}$ with non-negligible advantage, $\A$ can break the security of $\prf$.
		 \\
		 $\Rightarrow \hyb_1\equiv\hyb_0$.
		 \item[$\hyb_{2}$:] It is identical to $\hyb_1$ except that $\delta_{s}$ are randomly generated.
		 Since $\delta_s$ in $\hyb_1$ experiment are computed by $h$, which are generated by $\prf$.
		 Thus, any PPT $\A$ cannot distinguish $\delta_s$ and $\Tilde{\delta}_{s}$, given the security of $\prf$.
		If $\A$ has the non-negligible advantage to guess the real $s$, then the $\A$ can distinguish $\Tilde{\delta}_{s}$ and $\delta_s$ with non-negligible probability, which breaks $\prf$.
		\\
		 $\Rightarrow \hyb_2\equiv\hyb_1$.
		 \item[$\hyb_{3}$:] It is identical to $\ideal^{1^\lambda, \A}_{\Scal, \F}$, \ie, all the $\view_0,\view_1$ 
		 are simulated by $\Scal_0,\Scal_1$.
		 The randomness of $b, \l u \r_0, \l bu \r_0,\delta_s, \allowbreak\delta_{\l x\r_0},\delta_{\l x\r_1},\allowbreak\delta_x$ 
		 guarantees no non-negligible $\A$'s advantage of distinguishability to $\Scal_0$'s view.
		 Similarly, the randomness of $\l bu \r_1, \l u \r_1, \allowbreak\delta_{\l x\r_1}, \delta_s$ guarantees no non-negligible $\A$'s advantage of distinguishability to $\Scal_1$'s view.
		 Now, $\pp_0$ cannot obtain $\l x\r_1 $, while $\pp_1$ cannot obtain $s,\l x\r_0$ since $\{\Scal_i\}_{i\in\{0,1\}}$ cannot obtain $\{\Scal_{1-i}\}_{i\in\{0,1\}}$'s inputs using $\{\view_i\}_{i\in\{0,1\}}$.\\
		 $\Rightarrow \hyb_3\equiv\hyb_2$.
	 \end{itemize}
So,  $\real^{1^\lambda,\A}_{\Pi_\SM}\equiv\ideal^{1^\lambda, \A}_{\Scal, \F}$ for any PPT semi-honest $\A$.	
\end{proof}

\subsection{Security of 
\texorpdfstring{$\prosmm$}{π(SM)2}
}
Correctness has been checked using theoretical foundation for sparse-matrix in Appendix~\ref{sec:matrix_found_sparse}.
Functionality~\ref{func::mmult} defines arbitrary-sparse matrix multiplication without 
decomposition.

\begin{functionality}[!t]
	\caption{$\F_{\text{\osmm}}$: Ideal Functionality of $\prosmm$}\label{func::mmult}
\begin{algorithmic}[1]
	\REQUIRE $\adjmat\in \Mbb_{m,n}(\Rcal), \feamat\in \Mbb_{n,d}(\Rcal)$.
	\ENSURE $\l \adjmat\feamat \r_0, \l \adjmat\feamat \r_1$.
	\STATE Compute and generate random shares of $\adjmat\feamat$
	\RETURN $\l \adjmat\feamat \r$
\end{algorithmic}
\end{functionality}

\begin{theorem}[Security of $\prosmm$]
\label{secthe::smm}
Let $\Asf$ be a sparse 
matrix and $\Xsf$ be any vector/matrix.
The protocol $\prosmm$ realizes the functionality $\F_{\text{\osmm}}$ 
against static semi-honest adversaries.
\end{theorem}

\begin{proof}
The $\prosmm$ protocol sequentially call the independent subroutines of $5\ {\Pi}_{\ssp}, 2\ \Pi_\SM$, and $1\ \promult$ protocols that have been proved to be semi-honest secure.
The sequential composition theorem~\cite{joc/Canetti00} guarantees that security is closed under composition.
So, $\prosmm$ is semi-honest secure.
\end{proof}
 
\subsection{Security of 
\texorpdfstring{$\cgnn$}{\textcgnn}
}

\begin{theorem}
\label{secthe::sec_swan}
\cgnn securely realizes the functionality of GCN (Figure~\ref{graph_imple}) against static semi-honest adversaries.
\end{theorem}

\begin{proof}
    $\cgnn$ integrates the semi-honest protocols for all elementary operations like graph convolution and activation layers.
To obtain the secure inference or training protocol, we can sequentially compose the relevant protocols.
Correctness and 
security of private inference or training follow the integration of underlying sub-protocols.
By the sequential composition theorem~\cite{joc/Canetti00}, 
$\cgnn$ is semi-honest secure.
\end{proof}

\end{document}